\documentclass[12pt]{article}

\pdfoutput=1

\usepackage{amsmath}
\usepackage{graphicx,psfrag,epsf}
\usepackage{enumerate}
\usepackage{natbib}
\usepackage{verbatim}

\usepackage{url} 

\usepackage{amssymb}
\usepackage{amsfonts}
\usepackage{amsthm}

\usepackage{float}
\usepackage{color}

\newtheorem{assumption}{Assumption}
\newtheorem{thm}{Theorem}
\newtheorem{lem}{Lemma}
\newtheorem{cor}[thm]{Corollary}
\newtheorem{deff}{Definition}
\newtheorem{algorithm}{Algorithm}

\makeatletter
\newcommand{\pushright}[1]{\ifmeasuring@#1\else\omit\hfill$\displaystyle#1$\fi\ignorespaces}
\newcommand{\pushleft}[1]{\ifmeasuring@#1\else\omit$\displaystyle#1$\hfill\fi\ignorespaces}
\makeatother

\DeclareMathOperator*{\argmax}{arg\,max}
\DeclareMathOperator*{\argmin}{arg\,min}

\DeclareMathOperator*{\Var}{Var}

\def\bSig\mathbf{\Sigma}

\newcommand{\p}{^\prime}

\newcommand{\eps}{\varepsilon}
\newcommand{\diag}{\mathrm{diag}}

\newcommand{\tT}{\widetilde{T}}

\newcommand{\tY}{\widetilde{Y}}

\newcommand{\fE}{\mathfrak{E}}

\newcommand{\cA}{\mathcal{A}}

\newcommand{\cD}{\mathcal{D}}

\newcommand{\cF}{\mathcal{F}}

\newcommand{\cK}{\mathcal{K}}

\newcommand{\cN}{\mathcal{N}}

\newcommand{\cW}{\mathcal{W}}
\newcommand{\cX}{\mathcal{X}}

\newcommand{\bracket}[1]{\left[ #1 \right]}

\newcommand{\suma}{\sum_{a=1}^m}
\newcommand{\sumi}{\sum_{i=1}^n}

\newcommand{\fracn}{\frac{1}{n}}

\newcommand{\SAPE}{\mathrm{SAPE}}
\newcommand{\CMSE}{\mathrm{CMSE}}
\newcommand{\PAPE}{\mathrm{PAPE}}

\newcommand{\blind}{1}

\begin{document}

\def\spacingset#1{\renewcommand{\baselinestretch}%
{#1}\small\normalsize} \spacingset{1}

\addtolength{\textheight}{1.5in}%


\if1\blind
{
  \title{\bf Balanced Policy Evaluation and Learning for Right Censored Data}
  \author{Owen E. Leete\thanks{
    The authors gratefully acknowledge the following funding sources: the first author was funded in part by NIEHS grant T32 ES007018, second author was funded in part by NSF grant IIS-1846210, third author was funded in part by NIAID grants P30 AI050410 and R37 AI029168, fourth author was funded in part by NIAID grant P30 AI050410, and the last author was funded in part by NCI grant P01 CA142538.
    }\hspace{.2cm}\\
    Department of Biostatistics, \\ University of North Carolina at Chapel Hill\\
    and \\
    Nathan  Kallus \\
    School of Operations Research and Information Engineering,\\ Cornell University\\
    and \\
    Michael G. Hudgens \\
    Department of Biostatistics, \\ University of North Carolina at Chapel Hill\\
    and \\
    Sonia Napravnik \\
    Division of Infectious Diseases, Department of Medicine, \\ University of North Carolina at Chapel Hill \\
    and \\
    Michael R. Kosorok \\
    Department of Biostatistics, \\ University of North Carolina at Chapel Hill}
  \maketitle
} \fi

\if0\blind
{
  \bigskip
  \bigskip
  \bigskip
  \begin{center}
    {\LARGE\bf Balanced Policy Evaluation and Learning for Right Censored Data}
\end{center}
  \medskip
} \fi

\bigskip
\begin{abstract}
Individualized treatment rules can lead to better health outcomes when patients have heterogeneous responses to treatment. Very few individualized treatment rule estimation methods are compatible with a multi-treatment observational study with right censored survival outcomes. In this paper we extend policy evaluation methods to the right censored data setting. Existing approaches either make restrictive assumptions about the structure of the data, or use inverse weighting methods that increase the variance of the estimator resulting in decreased performance. We propose a method which uses balanced policy evaluation combined with an imputation approach to remove right censoring. We show that the proposed imputation approach is compatible with a large number of existing survival models and can be used to extend any individualized treatment rule estimation method to the right censored data setting. We establish the rate at which the imputed values converge to the conditional expected survival times, as well as consistency guarantees and regret bounds for the combined balanced policy with imputation approach. In simulation studies, we demonstrate the improved performance of our approach compared to existing methods. We also apply our method to data from the University of North Carolina Center for AIDS Research HIV Clinical Cohort.
\end{abstract}

\noindent%
\vfill

\newpage
\section{Introduction}
\label{sec:intro}

Differences in individual patient characteristics can result in significant heterogeneity in response to treatment. An individualized treatment rule (ITR) is a function which takes patient specific characteristics and recommends a treatment. The optimal ITR recommends the treatment that maximizes the benefit with respect to some clinical outcome. Treatment decisions are made based on published treatment guidelines which often list several available treatments from which the physician needs to choose based on expert opinion and personal experience. Using data driven approaches to help inform decision making can formalize the process and improve patient outcomes. There has been a considerable amount of effort devoted to the estimation of the optimal ITR \citep{qian2011,zhang2012,OWL,zhao2014,cui2017,kallus2017,athey2017}. 

To motivate our approach to estimating ITRs, we consider data from an observational cohort study of HIV+ patients who take different types of antiretroviral therapy (ART). The nature of this data presents several challenges to estimating ITRs. First, there are many treatment options. HIV infection is treated with combination therapies consisting of several drugs from multiple drug classes. Second, the data are frequently subject to right censoring. One measure of the effectiveness of treatment is the durability of an ART regimen, which measures how long a patient remains on the same ART. The durability is predictive of long-term patient morbidity and mortality, but it is prone to loss to follow-up. Third, the observational nature of the data means that the probability of a patient receiving a particular treatment may be related to measured and unmeasured factors including comorbid conditions and other therapies. This confounding by indication can make it difficult to establish causal relationships. 

Some early methods for estimating ITRs involved a two step approach that first uses regression methodology to estimate a conditional mean for the response under each treatment and then recommends the treatment with the best estimated conditional mean value \citep{qian2011}. Since these methods are based on regression modeling, they can be easily extended to meet the specific requirements of a wide array of data sources. However, there are a number of concerns about the performance of these methods \citep{beygelzimer2009,OWL}. Much of the recent work in ITR estimation is based on outcome weighted learning (OWL), which reformulates the estimation procedure as weighted classification \citep{OWL}. The original paper assumed fully observed, binary treatment data from a clinical trial. Several extensions have relaxed each of these assumptions individually \citep{zhao2014,wang2016,cui2017,liang2018}, but none of them relax all three assumptions simultaneously. Policy evaluation methods are 
designed to work with multiple treatments and observational data \citep{dudik2011,kallus2017,athey2017}. However, current extensions of these methods to right censored data use multiple weighting approaches, which results in high variance estimators, or rely on strong assumptions.

In this article, we propose a new method which addresses the shortcomings in existing approaches with respect to our motivating example. We focus on balanced policy evaluation and learning, an approach proposed by \citet{kallus2017}. Weighting approaches to policy evaluation can suffer from high variance when the weights are large. Existing variance reduction methods result in a biased estimator, but the amount of bias introduced is unknown, making it difficult to optimize the bias-variance tradeoff. Balanced policy evaluation seeks to quantify the bias-variance tradeoff nonparametrically by finding weights which minimize a measure of the conditional mean squared error (CMSE). We extend the balanced policy approach to right censored data by developing an imputation approach which replaces the censored observations with the conditional expected mean survival time. This allows us to define weighted and doubly robust estimators for policy evaluation for two or more treatments under right censoring. By more explicitly quantifying the bias-variance tradeoff and reducing the impact of censoring via imputation, the proposed approach is able to reduce the variance, leading to more efficient policy evaluation and learning.

While we focus on a single policy evaluation method, the imputation approach we propose could be used to extend any ITR estimation method to work with right censored data because it results in an estimated fully observed outcome vector. We demonstrate that the proposed approach results in improved performance when compared to existing methods. Nonetheless, the performance of balanced policy evaluation and learning can suffer in the presence of high dimensional, potentially noisy covariate information. We discuss a number of ways to mitigate this issue in the presence of right censored data, and demonstrate their use in simulation studies.

The remainder of this article is organized as follows. In Section \ref{sec:meth}, we discuss the policy evaluation approach to estimating ITRs, and introduce an imputation approach for right censored data. In Section \ref{sec:theory}, we develop the theory for the proposed method. Simulation studies are presented in Section \ref{sec:simulations}. We also illustrate our method using data from the University of North Carolina Center for AIDS Research HIV Clinical Cohort in Section \ref{sec:application}. The article concludes with a discussion of future work in Section \ref{sec:disc}. Some technical results are provided in the Appendix.

\vspace{-.5mm}\section{Methods}
\label{sec:meth}

\vspace{-.5mm}\subsection{Problem Setup and Notation}

Before discussing policy evaluation in the presence of right censoring, we first introduce some notation and describe policy evaluation and learning adapted to the fully observed survival time setting. Let $X \in \cX$ be the observed subject-level covariate vector, where $\cX$ is a $d$-dimensional vector space, and $A$ is the treatment, where $A\in\cA = [m] = \{1,\dots,m\}$ comes from a finite action space. For each $a\in\cA$ define $\tT(a) \in [0,\infty)$ as the unrestricted failure time that would be observed under treatment $a$, and let $\widetilde{\mathbf{T}} = (\tT(1),\dots,\tT(m))$. In many studies, the unrestricted failure time $\widetilde{\mathbf{T}}$ is rarely observable due to the existence of a maximum follow-up time $\tau < \infty$. We instead consider $\mathbf{T}= (T(1),\dots,T(m))$, a truncated version of $\mathbf{\tT}$ at $\tau$, i.e., $T(a) = \min(\tT(a),\tau)$ for $a\in[m]$. For a covariate vector $x$, treatment $a$, and failure time $T(a)$ we define the mean-outcome function as $\mu_a(x) = E[T(a) \,|\, X = x]$. 

In most survival studies there are some subjects for which observation of the failure is precluded by the occurrence of censoring events. For any treatment $a$, consider a censoring time $C(a)$ which is independent of $T(a)$ given  $X$ and $A$. Define the observed time $Y(a)=\min(T(a),C(a))$ and the failure indicator $\Delta=\mathbf{I}(T(a)\leq C(a))$. Let the failure time $T$ and censoring time $C$ be generated from distributions with conditional survival functions $S(t \,|\, x,a)$ and $G(c \,|\, x,a)$, respectively. The censored data consist of $n$ observations of covariate, treatment, observed time, and failure indicator: $\cD_n = \{(X_i, A_i, Y_i(A_i), \Delta_i)\}_{i=1}^n$. We assume that the data are drawn iid from a fixed distribution: $(X,A,\mathbf{T},\mathbf{C}) \sim P$, and that we observe 
$\left(X,A,Y(A)\right)$. This definition of $\left(X,A,Y(A)\right)$ implies the causal consistency assumption of the potential outcomes framework and the stable unit treatment value assumption (SUTVA), which ensures that the observed outcome for one subject is unaffected by the treatment assignments of the other subjects \citep{rubin1986}. We also assume that the data satisfy conditional exchangeability:
\begin{assumption} \label{as:exchangeability}
$T(a)\perp A \,|\, X$ for all $a \in \cA$.
\end{assumption}
\noindent This is commonly referred to as the no unmeasured confounders assumption, where the treatment decision is not influenced by the outcomes except through their mutual relationship with the observed covariates $X$. In our setting, conditional exchangeability is equivalent to there being 
an unknown propensity function $\varphi$ which, given covariate vector $x$, assigns treatment $a$  with probability $\varphi_a(x) = \mathbb{P}(A=a|X=x)$. 

A policy is a map $\pi: \cX\mapsto \Theta^m$ from the covariate space to a probability vector in the m-simplex, $\Theta^m = \{p\in[0,1]^m : \suma p_a = 1\}$. For an observation with covariate vector $x$, the policy $\pi$ specifies that treatment $a$ should be administered with probability $\pi_a(x)$. There are two main goals in this framework. The first goal, known as policy evaluation, is to estimate the average outcome that would be observed if treatments were assigned according to the policy $\pi$. Because the implementation of a bad policy can be costly or dangerous, new policies are generally evaluated using historical observations \citep{thomas2015}. 
The second goal, known as policy learning, attempts to identify a policy $\pi^*$  which maximizes the expected outcome with respect to some reward. In this article we will assume that larger outcomes are preferable, and thus the reward would be $T$ or $\log(T)$.

To formally define policy evaluation, consider a policy $\pi$ for which we wish to estimate the sample-average policy effect (SAPE),
$$\SAPE(\pi) = \fracn \sumi\suma \pi_a(X_i)\mu_a(X_i).$$ 
The SAPE quantifies the average outcome that would have been observed had the treatments been assigned according to a given policy $\pi$. The SAPE is strongly consistent for the population-average policy effect (PAPE), $\PAPE(\pi) = \mathbb{E}\left[\SAPE(\pi)\right] = \mathbb{E}\left[\suma \pi_a(X_i)\mu_a(X_i)\right]$.
If $\pi^*$ is such that  $\pi^*_a(x) > 0 \iff a = \argmax_{s\in[m]} \mu_s(x)$, then $\pi^*$ is an optimal policy. That is, $\pi^*$ maximizes $\SAPE(\pi)$ over all functions $\cX\mapsto\Theta^m$. For an optimal policy $\pi^*$, $R(\pi) = \SAPE(\pi^*) - \SAPE(\pi)$ is the regret of $\pi$. 
In policy learning, we wish to find a policy $\hat{\pi}$ that minimizes the expected regret.

\subsection{Existing Approaches to Policy Evaluation} \label{sec:PE_approaches}

Current approaches to policy evaluation generally fall into a few categories which are based on mean outcome modeling, weighting methods, or a combination thereof. The regression-based estimator first fits regression models $\hat{\mu}_a$ of $\mu_a$ for all $a\in[m]$. These models are then used to estimate the $\SAPE$ in a plug-in fashion
\begin{align}
    \psi^{\text{Reg}}(\pi) = \fracn\sumi\suma \pi_a(X_i)\hat{\mu}_a(X_i). \label{eq:reg}
\end{align}
This approach assumes that the regression models for each treatment are correctly specified, which can be a strong assumption when the number of treatments is large \citep{OWL}.

Weighting-based estimators use covariate and treatment data to find weights, $W(\pi,X_{1:n},T_{1:n})$, that reweight the outcome data to make it look as though it were generated by the policy being evaluated. These estimators are generally of the form
\begin{align}
    \hat{\psi}_W(\pi) = \frac{1}{n} \sumi W_i T_i, \label{eq:vanilla_weighted}
\end{align}
where the weights are often normalized to sum to $n$. With weights $W$ and regression models $\hat{\mu}_a$ we can define a doubly robust (DR) estimator 
\begin{align}
    \hat{\psi}^\text{DR}_{W,\hat{\mu}} = \fracn\sumi\suma \pi_a(X_i)\hat{\mu}_a(X_i) + \fracn \sumi W_i(T_i - \hat{\mu}_{A_i}(X_i)), \label{eq:doubly_robust}
\end{align}
which can be thought of as denoising the weighted estimator by subtracting the conditional mean from $T$, or debiasing the regression-based estimator using the reweighted residuals.

A common weighting-based approach is inverse propensity weighting (IPW) \citep{imbens2000}. By noting that $\SAPE(\pi) = \mathbb{E}\left[\fracn\sumi \pi_{A_i}(X_i)T_i(A_i)/\varphi_{A_i}(X_i)\,|\, X_{1:n}, A_{1:n}\right]$, for $\hat{\varphi}$, an estimate of the unknown propensity function $\varphi$, one can construct weights
\begin{align}
    W_i^\text{IPW} (\pi) = \pi_{A_i}(X_i)/\hat{\varphi}_{A_i}(X_i), \label{eq:IPW}
\end{align} which gives rise to the estimator $\hat{\psi}^{\text{IPW}} = \hat{\psi}_{W^{\text{IPW}}(\pi)}$. Since the propensities appear in the denominator, small values of $\hat{\varphi}$ can lead $\hat{\psi}^{\text{IPW}}$ to have a large variance. Some ad-hoc solutions for reducing the variance exist, such as clipping, where $\hat{\varphi}_{A_i}(X_i)$ is replaced by $\max\{M,\hat{\varphi}_{A_i}(X_i)\}$ for some $M$, but these methods generally have the effect of increasing the bias \citep{li2015}.

\citet{kallus2017} proposed a weighting-based estimator which simultaneously controls both the bias and variance. For a generic weighted estimator $\hat{\psi}_W(\pi)$, consider the conditional mean squared error (CMSE),
\begin{align}
\CMSE(\hat{\psi},\pi) = \mathbb{E}\left[ \left(\hat{\psi} - \SAPE(\pi) \right)^2 \;\middle|\; X_{1:n},A_{1:n} \right]. \label{eq:CMSE}
\end{align}
Balanced policy evaluation takes the $\CMSE$ and decomposes it into squared bias and variance components. The bias term relies on the weights $W$ and the unknown mean outcome function $\mu$. Rather than directly estimating $\mu$, the balanced policy approach defines a worst case bias by choosing the mean outcome function from a bounded class of functions that maximizes the squared bias. Once the worst case bias term has been identified, the variance term is estimated, and balanced policy weights are found by minimizing the worst case squared bias and variance with respect to the weights $W$. Additional details for balanced policy evaluation will be provided in Section \ref{sec:BP_RCD}.

\subsection{Extensions to Right Censored Data}

Analyzing right censored data without accounting for the missing data results in bias due to censoring \citep{fleming2011}, thus consistent policy evaluation requires methods designed to correct for the bias due to censoring.

Regression-based methods can be extended to right censored data by using regression models designed for censored data. Common model choices include parametric and semiparametric models such as the accelerated failure time and Cox proportional hazards models \citep{kalbfleisch2011}. However, these models make restrictive assumptions about the structure of the failure time distributions which are often difficult to verify. Non-parametric models of the survival function have been proposed, such as the kernel conditional Kaplan-Meier estimate \citep{dabrowska1987} or random forest based methods \citep{hothorn2004,ishwaran2008,zhu2012,steingrimsson2016,cui2019} which define localized versions of simple survival models, such as the Kaplan-Meier estimator. Non-parametric models make less restrictive assumptions, but they are less efficient and unlikely to produce good results in small sample sizes.

Inverse propensity weighted methods can be extended to right censored data by looking only at the observed failure times, and weighting based on the probability of observation. Inverse probability of censoring weighting (IPCW), is motivated by an idea similar to IPW. Since $\mathbb{E}[\Delta\,|\, X,A,T] = G(T\,|\, X,A)$ we have that $\mathbb{E}[Y\Delta/G(Y\,|\, X,A)\,|\, X,A] = \mathbb{E}[T\,|\, X,A].$ IPCW can be combined with IPW by defining weights 
\begin{align}
W_i^\text{IPW,IPCW} (\pi) = \frac{\pi_{A_i}(X_i)\Delta_i}{\hat{\varphi}_{A_i}(X_i)\hat{G}(Y_i\,|\, X_i,A_i)}. \label{eq:IPW_IPCW}
\end{align}
As long as the models for $\hat{\varphi}$ and $\hat{G}$ are correctly specified, the resulting estimator, $\hat{\psi}_W^\text{IPW,IPCW}(\pi)$, will be consistent for the SAPE \citep{anstrom2001}. However, with two weights in the denominator, the issues with high variance will be compounded.

Balanced policy weights are incompatible with IPCW. Depending on when the IPC  weights are applied,
balanced policy weights are either no longer identifiable, or they lose the property of minimizing the worst case bias and variance. To extend balanced policy to right censored data, we use an imputation approach similar to one proposed by \cite{cui2017}. Consider an imputed failure time vector
$\tY_i \equiv 
\Delta_i Y_i + (1-\Delta_i)E\bracket{T_i\;\middle|\;X_i,A_i,Y_i,T_i>Y_i},$
where the censoring times are replaced with the conditional expected failure time given the covariates, treatment, and the knowledge that the subject survived at least until time $C_i$. 


Estimating the expected failure time for the censored subjects uses a conditional estimate of the survival function.
In this article we focus on estimation of the conditional expected failure time using random forest methods for right censored data because they provide nonparametric estimates with good performance for higher dimensional data \citep{cui2019}. Random survival forests (RSF) \citep{ishwaran2008} and recursively imputed survival trees (RIST) \citep{zhu2012} are based on extensions of random forests (RF) \citep{breiman2001} and extremely randomized trees (ERT) \citep{geurts2006}, respectively, where the splitting rule for each node is chosen to maximize the log-rank test statistic. \cite{cui2019} note that a splitting rule based on the standard log-rank test statistic can produce biased results because it does not appropriately control for the censoring distribution, and instead propose a splitting rule which corrects for the bias due to censoring. 

Regardless of the method used to estimate the survival function, given $\widehat{S}(\cdot\,|\, X,A)$, the conditional expected survival time can be estimated as
\begin{align*}
\widehat{E}\bracket{T_i\;\middle|\;X_i,A_i,Y_i,T_i>Y_i} = \frac{\int_Y^\tau t\, d\widehat{F}(t\,|\, X,A) + \tau \widehat{S}(\tau\,|\, X,A)}{\widehat{S}(Y\,|\, X,A)},
\end{align*}
and the estimated imputed failure time vector can be defined as
\begin{align}
\widehat{Y}_i = \Delta_i Y_i + (1-\Delta_i)\widehat{E}\bracket{T_i\;\middle|\;X_i,A_i,Y_i,T_i>Y_i}. \label{eq:hatY}
\end{align}
By using $\widehat{Y}$ in place of $T$ in equation \eqref{eq:vanilla_weighted} we can now define a balanced policy approach to right censored data.

\subsection{Balanced Policy Evaluation with Right Censored Data} \label{sec:BP_RCD}

Recall from equation \eqref{eq:CMSE} that we defined the CMSE as the squared difference between a weighted estimator $\psi_W(\pi)$ and the $\SAPE(\pi)$. We define
\begin{align*}
B_a(W,\pi_a;\mu_a) = \fracn \sum_{i=1}^n\left(W_i\delta_{A_ia} - \pi_a(X_i)\right)\mu_a(X_i) \quad\text{ and }\quad B(W,\pi;\mu) &= \sum_{a=1}^m B_a(W,\pi_a;\mu_a),
\end{align*} 
where $\delta_{ij}=\mathbf{I}[i=j]$ is the Kronecker delta. Letting $\eps_i = T_i - \mu_a(X_i)$, from \citet{kallus2017} Theorem 1, we have that
$$\hat{\psi}_W - \SAPE(\pi) = B(W,\pi;\mu) + \fracn \sum_{i=1}^n W_i\eps_i,$$ and, under assumption \ref{as:exchangeability},
\begin{align*}
\CMSE(\hat{\psi}_W,\pi) = B^2(W,\pi;\mu) + \frac{1}{n^2} W^T\Sigma W,
\end{align*}
where $\Sigma = \diag\left(\mathbb{E}[\eps_i^2\,|\, X_i,A_i]_{i=1}^n\right)$. 
This decomposes the CMSE into squared bias and variance components. The measure of the variance will be given by the norm of weights $W^T\Lambda W$ for a positive semidefinite (PSD) matrix $\Lambda$. Instead of directly estimating the bias term, we find the worst case bias by maximizing $B^2(W,\pi;f)$ over all functions $f\in\cF$ for a bounded function class $\cF$.

To define the worst case bias, we look at functions, $f$, in the unit ball of a direct product of reproducing kernel Hilbert spaces (RKHS):
\begin{align}
\|f\|_{p,\cK_{1:m},\gamma_{1:m}} = \left( \sum_{a=1}^m \|f_a\|^p_{\cK_a} / \gamma_a^p \right)^{1/p}, \label{eq:norm_of_f}
\end{align}
where $\| \cdot \|_{\cK_a}$ is the norm of the RKHS given by the PSD kernel $\cK(\cdot,\cdot): \cX^2 \mapsto \mathbb{R}$. Any choice of $\|\cdot\|$ 
gives rise to the worst case CMSE objective:
\begin{align*}
\mathfrak{E}^2(W,\pi;\| \cdot\|,\Lambda) = \sup_{\|f\|\leq1} B^2(W,\pi;f) + \frac{1}{n^2} W^T\Lambda W.
\end{align*}

Balanced policy evaluation is given by the estimator $\hat{\psi}_{W^*(\pi;\| \cdot\|,\Lambda)}$ where $W^*(\pi) = W^*(\pi;\| \cdot\|,\Lambda)$ is the minimizer of $\fE^2(W,\pi;\| \cdot\|,\Lambda)$ over the space of all weights $W$ that sum to $n$, $\cW= \{ W\in \mathbb{R}^n_+: \sum_{i=1}^n W_i=n\} =n \Theta^n$. Specifically,
\begin{align}
W^*(\pi;\| \cdot\|,\Lambda) \in \argmin_{W\in\cW} \fE^2(W,\pi;\| \cdot\|,\Lambda). \label{eq:BP_weights}
\end{align}
With the imputed survival time vector, $\widehat{Y}$, and the balanced policy weights, we can define the weighted and doubly robust balanced policy estimators for right censored data:
\begin{align}
    \widehat{\psi}_{W^*,\widehat{Y}} &= \fracn \sumi W^*\widehat{Y}_i, \label{eq:BP_Estimator} \\
    \widehat{\psi}^\text{DR}_{W^*,\widehat{\mu}}(\pi) &= \fracn\sumi\suma \pi_a(X_i)\widehat{\mu}_a(X_i) + \fracn \sumi W^*_i(\widehat{Y}_i - \widehat{\mu}_{A_i}(X_i)). \label{eq:BP_DR_Estimator}
\end{align}


Next we consider policy learning. For a given policy class $\Pi\subset[\cX\mapsto\Theta^m]$, let $\hat{\pi}\in\Pi$ be the policy which maximizes the balanced policy evaluation estimator, $\hat{\phi}_W$, using the weights defined in equation \eqref{eq:BP_weights}. We formulate this as a bilevel optimization problem:
\begin{align}
    \hat{\pi}^\text{balanced} \in \argmax_\pi\left\{ \hat{\psi}_W :\pi\in\Pi,W\in\argmin_{W\in\cW} \fE^2(W,\pi;\|\cdot\|,\Lambda) \right\}, \label{eq:learner} \\
    \hat{\pi}^\text{balanced-DR} \in \argmax_\pi\left\{ \hat{\psi}_{W,\hat{\mu}} :\pi\in\Pi,W\in\argmin_{W\in\cW} \fE^2(W,\pi;\|\cdot\|,\Lambda) \right\}. \label{eq:learner_DR}
\end{align}

It is difficult to maximize the policy $\pi$ with respect to the full policy class $\Pi=[\cX\mapsto\Theta^m]$, so it is common to use a reduced class such as the parameterized policy class 
\begin{align}
    \Pi_\text{logit} = \{\pi_a(x;\beta_a) \propto \exp(\beta_{a0} + \beta_a^Tx)\}. \label{eq:pi_logit}
\end{align} 
Using a reduced policy class limits the flexibility of balanced policy learning, but for moderate sample sizes the reduced complexity can allow for more accurate estimation of $\pi^*$.

A high-level description of the proposed method is given in Algorithm \ref{alg:RCBP} below. 

\begin{algorithm} \label{alg:RCBP}
\emph{Pseudo algorithm for the Balanced Policy Evaluation and Learning with Imputation}

\noindent{\bf Step 1:} Use $\{(X_i, A_i , Y_i , \Delta_i \}^n_{i=1}$ to fit a model for the conditional survival function $S(t\,|\, X,A)$, and create the imputed survival time vector, $\widehat{Y}$, by using $\widehat{S}(t\,|\, X,A)$ to estimate the conditional expected survival times $\widehat{E}\bracket{T_i\;\middle|\; X_i,A_i,Y_i,T_i>Y_i}$ for the censored observations. \\
\noindent{\bf Step 2:} For a norm, $\|\cdot\|_{p,\cK_{1:m},\gamma_{1:m}}$, and diagonal matrix, $\Lambda$, define  a policy class $\Pi$, and select $\pi_0\in\Pi$ \\
\noindent{\bf Step 3:} Find the balanced policy weights associated with $\|\cdot\|_{p,\cK_{1:m},\gamma_{1:m}}$, $\Lambda$, and $\pi_i$ and use the weights $W^*$ to calculate $\widehat{\psi}_{W^*}(\pi_i)$ according to equation  \eqref{eq:BP_Estimator}. \\
\noindent{\bf Step 4:} Calculate the gradient of the policy, $\pi_i$, with respect to $\widehat{\psi}_{W^*}(\pi_i)$ and use gradient ascent to find an improved policy $\pi_{i+1} \in \Pi$ \\
\noindent{\bf Step 5:} Repeat steps 3-4 until convergence to an estimated optimal policy $\widehat{\pi}$.
\end{algorithm}


\subsection{High Dimensional Considerations}

Because the balanced policy weights are defined using a kernel norm $\|\cdot\|_\cK$, which depends only on the sample covariate space $\mathcal{X}$, balanced policy evaluation and learning can be sensitive to the presence of covariates which contain redundant or noisy information. The negative 
effects of high-dimensional noisy data can be reduced in several ways including the choice of the kernel norm, and feature elimination, both of which we detail here.


Recall from equation \eqref{eq:norm_of_f} that we defined the norm of a function with respect to the direct product of RKHSs. One commonly used kernel is the Mahalanobis radial basis function (RBF) kernel
\begin{equation}
    \cK_s (x,x\p) = \exp(-(x-x\p)^T\widetilde{\Sigma}^{-1}(x-x\p)/s^2). \label{eq:mRBF}
\end{equation} In Lemma 1 point (3) of \citet{kallus2017}, it was shown that, for the norm defined in \eqref{eq:norm_of_f}, if $p = 2$ and $f_a$ has a Gaussian process prior, then kernel hyperparameters, such as $s$ and $\widetilde{\Sigma}$, can be selected using the marginal likelihood principle. Choosing the hyperparameters in this manner 
reduces the influence of the unimportant predictors. In simulations this has greatly improved the accuracy of the balanced policy estimator with respect to evaluation. 
However, since the redundant or noisy predictors are not eliminated, 
the dimension of the search space remains high, which can still adversely impact policy learning. 

Reducing the dimension of the search space using variable selection (VS) methods can improve the performance of policy learning. Nearly any VS method could be used, but to keep our method as general as possible we examined non-parametric VS methods such as the risk-recursive feature elimination algorithm (risk-RFE) proposed by \citet{dasgupta2019}. Risk-RFE uses kernel machines to rank the importance of the features based on the regularized risk. Once the features are ranked in order of estimated importance, change point estimation methods can be used to identify the important predictors based on the increase in regularized risk at each step. For additional details see \citet{dasgupta2019}.

In our implementation for high dimensional data, we used methods where tuning hyperparameters and variable selection required a complete outcome vector. Therefore, these methods were implemented prior to Step 3 of Algorithm \ref{alg:RCBP}, but methods which are compatible with right censoring could be implemented prior to Step 1.%
\section{Theory}%
\label{sec:theory}%
The properties of the proposed estimator will depend on the method used to estimate $\widehat{Y}$, so to facilitate the discussion we make the following assumption:%
\begin{assumption}[Convergence rate]%
There exists some sequence $r_n\rightarrow 0$ such that%
 \begin{align*}%
     \sup_{t<\tau}E_{X,A}\left|\widehat{S}_n(t\,|\, X,A) - S(t\,|\, X,A)\right| = O_p(r_n).%
 \end{align*}%
 \label{as:convergence_Y}%
\end{assumption} \vspace*{-10mm}
\noindent 
Assumption \ref{as:convergence_Y} may seem strong, but it is met for a large number of methods 
for estimating $S(t\,|\, X,A)$. For example, the assumption holds for AFT and Cox PH models with $r_n=1/\sqrt{n}$. Assumption \ref{as:convergence_Y} is also met for non-parametric methods, such as kernel conditional Kaplan-Meier and random forest based methods. For the kernel conditional Kaplan-Meier, $r_n=
\log(n^{d/(d+4)})^{1/2}/n^{2/(d+4)}$, where $d$ is the covariate dimension \citep{dabrowska1989}. The convergence rate for the bias corrected random survival forest method proposed by \citet{cui2019} depends on the splitting criteria, but the rate for the theoretically optimal splitting criteria is $r_n=n^{-1/(d+2)}$. Convergence rates for other random forest methods \citep{ishwaran2008,zhu2012} have not yet been established, but we believe that the rates will be similar to those found in \citet{cui2019}.

We show that a convergence rate for $\widehat{S}_n(t\,|\, X,A)$ is enough to guarantee the convergence rate of $\widehat{Y}$; however, Theorems 12.4 and 12.5 of \citet{kosorok2008} imply that it is also sufficient to have a convergence rate for an estimate of the cumulative hazard function, $\widehat{\Lambda}(t\,|\, X,A)$.

\begin{lem}
If there exists some sequence $r_n\rightarrow 0$ such that,
\begin{align*}
     \sup_{t<\tau}E_{X,A}\left|\widehat{S}(t\,|\, X,A) - S(t\,|\, X,A)\right| &= O_p(r_n)
\end{align*}
then
\begin{align*}
     \left|\fracn\sumi (1-\Delta_i) W^*_i(\widehat{Y}_i-T_i)\right| &= O_p(r_n).
\end{align*} \label{thm:Preservaton_of_rates}
where $W^*=W^*(\pi;\| \cdot\|,\Lambda)$ are the balanced policy weights.
\end{lem}
To prove Lemma \ref{thm:Preservaton_of_rates} we need an additional assumption.
\begin{assumption} \label{as:leave_one_out}
Let $\widehat{Y}^{(i)}_i = \widehat{E}^{(i)}_{n-1}(T\,|\, X_i,A_i,Y_i,\Delta_i),$
where $\widehat{E}^{(i)}_{n-1}$ is based on all observations except for the $i^\text{th}$ one. For all $n\geq1$ and any $1\leq i\leq n$, $$E\left[(1-\Delta_i)(\widehat{Y}_i-\tY_i)^2\right]\leq E\left[(1-\Delta_i)(\widehat{Y}^{(i)}_i-\tY_i)^2\right].$$%
\end{assumption}
\noindent Since $\widehat{Y}_i$ is based on more information than $\widehat{Y}^{(i)}_i$, this assumption should be met provided $\widehat{E}_n$ is consistent. The proof of Lemma \ref{thm:Preservaton_of_rates} can be found in the appendix.

Consistent evaluation requires a weak form of overlap between the unknown propensity function and the target policy being evaluated:
\begin{assumption}[Weak overlap]
$\mathbb{P}(\varphi_a(X)>0 \vee \pi_a(X)=0) = 1$ for all $a\in[m]$, and \\ $E[\pi^2_A(X)/\varphi^2_A(X)]<\infty.$ \label{as:weak_overlap}
\end{assumption}
\noindent Another requirement is well-specification of the mean-outcome function. For balanced policy evaluation the mean-outcome function is well-specified if it is in the RKHS product used to compute $W^*(\pi)$. Otherwise, consistency is also guaranteed if the RKHS product consists of $C_0$-universal kernels, defined below, such as the RBF kernel. 
\begin{deff} A PSD kernel $\cK$ on a Hausdorff $\cX$ (e.g., $\mathbb{R}^d$) is \emph{$C_0$-universal} if, for any continuous function $g : \cX \mapsto \mathbb{R}$ with compact support (i.e., for some compact $C$, $\{x : g(x) =
0 \subseteq C\}$) and $\eta>0$, there exists $n\p,\alpha_1,x_1,\dots,\alpha_{n\p},x_{n\p}$ such that $\sup_{x\in\cX} \left| \sum_{j=1}^{n\p} \alpha_j\cK(x_j,x) - g(x) \right|\leq\eta$.
\end{deff}
\begin{thm} Fix $\pi$ and let $W_n^*(\pi) = W_n^*(\pi;\|f\|_{p,\cK_{1:m},\gamma_{n,1:m}},\Lambda_n)$ with $0\prec\underline{\kappa}I\preceq\Lambda_n\preceq\overline{\kappa}I$, $0<\underline{\gamma}\leq \gamma_{n,a}\leq \overline{\gamma}\;\, \forall a\in[m]$ for each $n$. Suppose Assumptions \ref{as:exchangeability}-\ref{as:weak_overlap} hold, $\Var(T|X)$ is a.s. bounded, $\mathbb{E}[\sqrt{\cK_a(X,X)}]<\infty$, and $\mathbb{E}[\cK_a(X,X)\pi^2_A(X)/\varphi^2_A(X)]<\infty$. Then the following two results hold:
\begin{itemize}
    \item If $\|\mu_a\|_{\cK_a}<\infty$ for all $a \in [m]$, then \hfill $\hat{\psi}_{W^*_n(\pi),\widehat{Y}} - \SAPE(\pi) = O_p(r_n+1/\sqrt{n})$,
    \item If $\cK_a$ is $C_0$-universal for all $a \in [m]$, then \hfill $\hat{\psi}_{W^*_n(\pi),\widehat{Y}} - \SAPE(\pi) = o_p(1)$.
\end{itemize} \label{thm:VanillaRates}
\end{thm}

\begin{cor}
Suppose the assumptions of Theorem \ref{thm:VanillaRates} hold, then
\begin{itemize}
    \item If $\|\mu_a\|_{\cK_a}<\infty$ and $\|\widehat{\mu}_{na}\|_{\cK_a}=O_p(1)$ for all $a \in [m]$, then  \\ \textcolor{white}{.} \hfill $\hat{\psi}_{W^*_n(\pi),\widehat{\mu}_n,\widehat{Y}} - \SAPE(\pi) = O_p(r_n+1/\sqrt{n})$,
    \item If $\cK_a$ is $C_0$-universal for all $a \in [m]$, then \hfill $\hat{\psi}_{W^*_n(\pi),\widehat{\mu}_n,\widehat{Y}} - \SAPE(\pi) = o_p(1)$.
\end{itemize} \label{cor:DRRates}
\end{cor}

\begin{proof}[Proof of Theorem \ref{thm:VanillaRates}]
Define $$\hat{\psi}_{W^*_n(\pi),T} = \fracn \sumi W^*_n(\pi)T_i, \quad\text{and}\quad \hat{\psi}_{W^*_n(\pi),\hat{Y}} = \fracn \sumi W^*_n(\pi)\hat{Y_i},$$ as the estimators which use the true and estimated imputed failure time vectors respectively. We can then decompose $\hat{\psi}_{W^*_n(\pi),\widehat{Y}}$ as follows:
\begin{align*}
    \hat{\psi}_{W^*_n(\pi),\widehat{Y}} - \SAPE(\pi) 
    &= \fracn\sumi (1-\Delta_i) W_i(\widehat{Y}_i-T_i) + \hat{\psi}_{W^*_n(\pi),T} - \SAPE(\pi).
\end{align*}
By Lemma \ref{thm:Preservaton_of_rates} we have that $$\fracn\sumi (1-\Delta_i) W_i(\widehat{Y}_i-T_i) = O_p(r_n).$$
If $\|\mu_a\|_{\cK_a}<\infty$ for all $a \in [m]$, then by \citet{kallus2017} Theorem 3, we have that
\begin{align*}
    \hat{\psi}_{W^*_n(\pi),T} - \SAPE(\pi) &= O_p(1/\sqrt{n}), \\
    \text{and thus } \quad \hat{\psi}_{W^*_n(\pi),\widehat{Y}} - \SAPE(\pi) &= O_p(r_n+1/\sqrt{n}). \qquad
\end{align*}

If $\cK_a$ is $C_0$-universal for all $a \in [m]$, then by \citet{kallus2017} Theorem 3, we have that
\begin{align*}
\hat{\psi}_{W^*_n(\pi),T} - \SAPE(\pi) &= o_p(1), \\
\text{and thus } \quad \hat{\psi}_{W^*_n(\pi),\widehat{Y}} - \SAPE(\pi) &= o_p(1). \tag*{\qedhere}
\end{align*}
\end{proof}
\noindent The proof of Corollary \ref{cor:DRRates} follows immediately from Lemma \ref{thm:Preservaton_of_rates}, as well as Corollary 4 of \citet{kallus2017} in a manner similar to the proof of Theorem \ref{thm:VanillaRates}.

Next, we establish the consistency and learning rates of the balanced policy learner uniformly over policy classes. For a given policy class $\Pi$, we define the sample regret as
$$\widehat{R}_{\Pi}(\hat{\pi}) = \max_{\pi\in\Pi}\SAPE(\pi) - \SAPE(\hat{\pi}),$$
and the population regret as
$$R_{\Pi}(\hat{\pi}) = \max_{\pi\in\Pi}\PAPE(\pi) - \PAPE(\hat{\pi}).$$
Convergence of these regret quantities require that the best-in-class policy is learnable. We quantify learnability using Rademacher complexity, but the results could be extended to VC dimension. Let us define
$$\widehat{\mathfrak{R}}_n(\cF) = \frac{1}{2^n} \sum_{\rho_i\in\{-1,+1\}^n} \rho_if(X_i), \qquad \mathfrak{R}_n(\cF)=\mathbb{E}[\widehat{\mathfrak{R}}_n(\cF)].$$
If $\cF\subseteq [\cX\rightarrow\mathbb{R}^m]$, let $\cF_a = \{(f(\cdot))_a:f\in\cF\}$ and set $\mathfrak{R}_n(\cF) = \suma \mathfrak{R}_n(\cF_a)$ and $\widehat{\mathfrak{R}}_n(\cF) = \suma \widehat{\mathfrak{R}}_n(\cF_a)$. We also must use a stronger version of the overlap assumption.
\begin{assumption}[Strong overlap]
$\exists \alpha\geq1$ such that $\mathbb{P}(\varphi_a(X)>1/\alpha) = 1 \;\, \forall a\in[m]$. \label{as:strong_overlap}
\end{assumption}

\begin{thm}
Fix $\Pi\subseteq[\cX\rightarrow\Theta^m]$ and let $W_n^*(\pi) = W_n^*(\pi;\|f\|_{p,\cK_{1:m},\gamma_{n,1:m}},\Lambda_n)$ with $0\prec\underline{\kappa}I\preceq\Lambda_n\preceq\overline{\kappa}I$, $0<\underline{\gamma}\leq \gamma_{n,a}\leq \overline{\gamma}\;\, \forall a\in[m]$, $n$, and $\pi\in\Pi$. Suppose assumptions \ref{as:exchangeability}-\ref{as:convergence_Y} and \ref{as:strong_overlap} hold, $|\epsilon_i|\leq B $ a.s. bounded, and $\sqrt{\cK_a(X,X)}\leq \Gamma\; \forall a\in[m]$ for $\Gamma\geq 1$. If $\widehat{\pi}^\text{balanced}_n$ is as in \eqref{eq:learner}, then the following results hold:
\begin{itemize}
    \item If $\|\mu_a\|_{\cK_a}<\infty$ for all $a \in [m]$, then \hfill $R_\Pi(\widehat{\pi}^\text{balanced}_n) = O_p(\mathfrak{R}_n(\Pi)+r_n+1/\sqrt{n})$,
    \item If $\cK_a$ is $C_0$-universal for all $a \in [m]$, then \hfill $R_\Pi(\widehat{\pi}^\text{balanced}_n) = o_p(1)$.
\end{itemize} 
If $\widehat{\pi}^\text{balanced-DR}_n$ is as in \eqref{eq:learner_DR}, then the following results hold:
\begin{itemize}
    \item If $\|\mu_a\|_{\cK_a}<\infty$ and $\|\widehat{\mu}_{na}\|_{\cK_a}=O_p(1)$ for all $a \in [m]$, then \\ \textcolor{white}{.} \hfill $R_\Pi(\widehat{\pi}^\text{balanced}_n) = O_p(\mathfrak{R}_n(\Pi)+r_n+1/\sqrt{n})$,
    \item If $\cK_a$ is $C_0$-universal for all $a \in [m]$, then \hfill $R_\Pi(\widehat{\pi}^\text{balanced-DR}_n) = o_p(1)$.
\end{itemize} \label{thm:Learning_rate}
\end{thm}

The proof of Theorem \ref{thm:Learning_rate} follows directly from Assumption \ref{as:convergence_Y} and Lemma \ref{thm:Preservaton_of_rates}, as well as Corollary 7 of \citet{kallus2017}. All the same results hold when replacing $\mathfrak{R}_n(\Pi)$ with $\widehat{\mathfrak{R}}_n(\Pi)$ and/or replacing $R_{\Pi}$ with $\widehat{R}_{\Pi}$.

\section{Simulation Studies}
\label{sec:simulations}

We have conducted simulation studies to assess the performance of the proposed method in comparison with existing alternatives. The models we compare are the simple weighted and doubly robust versions of the balanced policy approach with the imputed failure time vector defined in equation \eqref{eq:hatY}, a normalized clipped version of the approach which combines IPW with IPCW as defined in equation \eqref{eq:IPW_IPCW}, and a normalized clipped version of the IPW approach combined with the imputed vector from \eqref{eq:hatY}. In all settings, both the optimal and estimated policies are members of the reduced policy class $\Pi_\text{logit}$ defined in \eqref{eq:pi_logit}. 
The kernel used to define the worst case bias in the balanced policy approach is the Mahalanobis RBF kernel in equation \eqref{eq:mRBF}, where the hyperparameters are chosen by the marginal likelihood method using Gaussian process regression. The same kernel is used for all of the treatments, $\|\cdot\|_{K_a}=\|\cdot\|_{K}$ for all $a\in[m]$. For each setting, the conditional expected survival times, $\hat{Y}$, are estimated using RIST with the tuning parameter settings suggested in \citet{zhu2012}. Propensity scores are estimated using a correctly specified Gaussian process classifier. Due to different censoring mechanisms, the weights used for IPCW are estimated differently for each setting. For the first setting the censoring weights are estimated using RIST with the recommended tuning parameter settings. For the second setting the censoring weights are estimated using the correctly specified accelerated failure time (AFT) model. A testing dataset with size 10000 is used to approximate the population regret $\left(\PAPE(\pi^*)- \PAPE(\hat{\pi})\right)$. For each setting the log survival time is used to estimate the optimal policy on simulated datasets of size $N=500, 1000$, and $2000$. Each simulation is repeated 100 times.

\subsection{Simulation Settings}

Our first setting is a slight modification of the setting found in \citet{kallus2017}. To begin, we sample the covariate vector from a multivariate normal distribution $X \sim \cN ({\mathbf 0},\Sigma_d)$, where $d=10$ and $\Sigma_d$ is compound symmetric covariance with diagonal elements of 1 and off diagonal elements of 0.2. The treatments are assigned with probabilities which depend on the first two elements of of $X$. Specifically $A\,|\, X\sim$ Multinoulli$(p_1(X),\dots,p_5(X))$ where $p_a(X)\propto N((X_1,X_2) - \overline{X}_a,I_{2\times 2})$ for $a\in[m]$, 
$\overline{X}_1 = (0,0)$, $\overline{X}_2 = (1,0)$, $\overline{X}_3 = (0,1)$, $\overline{X}_4 = (-1,0)$, and $\overline{X}_5 = (0,-1)$. The unrestricted failure times are drawn from a log-normal distribution where $\log(\tT)\,|\, A\sim N(\mu_A(X),\sigma^2=1)$ where $\mu_a(x) = \exp(1)-\exp(1-1/\|x-\chi_a\|_2)$, and $\chi_a = (Re,Im)(e^{-i2\pi a/5}/\sqrt{2})$ for $a \in [5]$, $\epsilon_i\sim N(0,1)$, and $(Re,Im)$ are the real and imaginary components respectively. The observed failure time $T = \min(\tau,\tT)$, where $\tau=3.5$. This mean outcome process results in an optimal treatment policy consisting of 5 equal sized wedges arranged radially about the origin. The censoring times are drawn from a log-normal distribution where $\log(C)\,|\, A\sim N(\mu_C(X),\sigma^2=2)$ where $\mu_C(x) = (5/2)-(1/2)\exp(1-1/\|X_1+X_2\|_2))$. The censoring rate is approximately 45\%.

In the second setting, the failure times are drawn from log-normal distributions where the treatment effect has both prognostic and prescriptive elements. The 10-dimensional covariate vector, $X$, is drawn from a uniform$(-1,1)$ distribution with a compound symmetric covariance structure by first drawing from a multivariate normal distribution  $\cN ({\mathbf 0},\Sigma_{d})$ where $\Sigma_{d}$ is the same as in setting 1. The CDF transformation was applied to the normally distributed covariates which where then scaled to be within -1 and 1. The treatments are assigned with probabilities which depend on the first three elements of of $X$, where $A\,|\, X\sim$ Multinoulli$(p_1(X),\dots,p_3(X))$ where $p_a(X)\propto N((X_1,X_2,X_3) - \overline{X}_a,I_{3\times 3})$ for $a\in[3]$, 
$\overline{X}_1 = (-0.5 , -0.5 ,  0.4)$, $\overline{X}_2 = ( 0, 0, -0.75)$, $\overline{X}_3 = (0.5, 0.5, 0.4)$. The unrestricted failure times are drawn from a log-normal distribution where $\log(\tT) = 0.2-0.6X_1+0.2X_2+0.4X_3+ \mu_a + \epsilon$, where $\mu_1 = 0.2X_1 - 0.3X_2$, $\mu_2 =  0.1 + 0.1X_3$,
$\mu_3 = -0.1 - 0.2X_1 + 0.4X_2 - 0.2X_3$, and $\epsilon\sim N(0,\sigma^2=1)$. The failure times $T = \min(\tau,\tT)$, where $\tau=1.5$. The censoring times are drawn from a log-normal distribution which depends on the assigned treatment, where $\log(C) = 0.6 - 0.4X_1 + 0.3X_2 + 0.8X_3 + \nu_a + \epsilon$, with $\nu_{1} = - 0.1X_1 - 0.2X_2$, $\nu_{2} = 0.1X_2$, $\nu_{3} = -0.1 + 0.3X_2 - 0.4X_3$ and $\epsilon \sim N(0,\sigma^2=1)$. The censoring rate is about 45\%.

\subsection{Simulation Results}

The results for setting 1 can be found in Figure \ref{fig:BP_Orig}. This setting is designed to have mismatch between the unknown propensity function and the optimal policy such that the inverse propensity weights are large, resulting in high variability for the IPW with IPCW, and IPW with Imputation approaches. As expected, the increased variance of the IPW based estimators results in reduced performance. The more stable balanced policy with imputation approach performs better than the other approaches. In this setting the feature elimination method effectively eliminates noise variables, even for small sample sizes, which increases performance over the methods using a higher dimensional feature space. We also examined doubly robust approaches for each of the methods, but the mean model was not accurate enough to improve performance over the simple weighted methods. The results for the doubly robust approaches can be found in the supplemental material.

 \begin{figure}[htb]
 \begin{center}
 \includegraphics[width=4.5in]{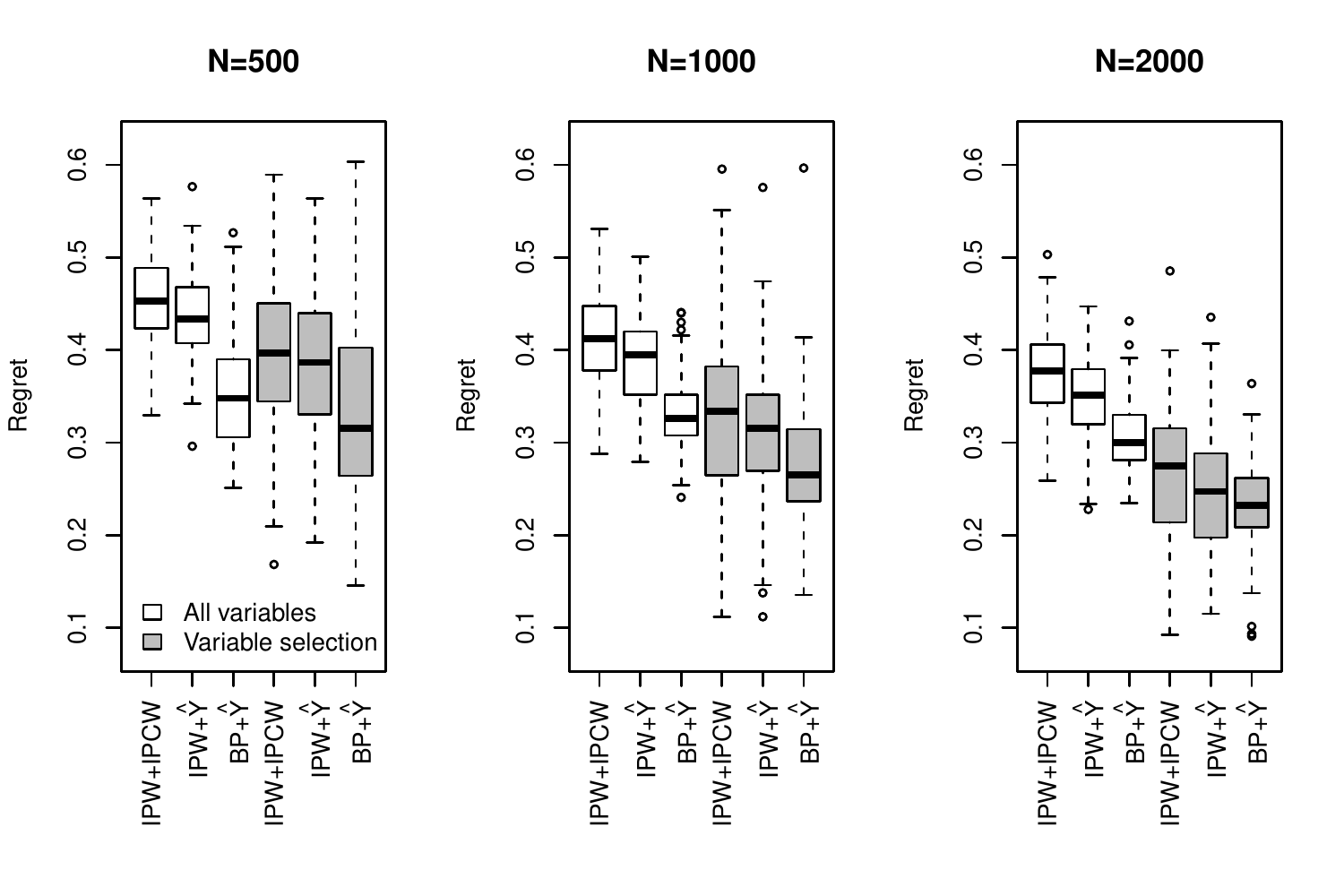}
 \end{center}
 \caption{Boxplots showing estimated population regret for setting 1. Smaller is better. IPW+IPCW: IPW combined with IPCW;  IPCW+$\hat{Y}$: IPW with imputed outcome vector; BP+$\hat{Y}$: Balanced policy with imputed outcome vector. Variable selection used riskRFE to remove variables. \label{fig:BP_Orig}}
 \end{figure}

The results for setting 2 can be found in Figure \ref{fig:AFT}. Again the higher variability of the IPW with IPCW, and IPW with Imputation approaches reduced performance when compared to the balanced policy approach. The doubly robust approaches for the balanced policy with imputation and IPW with Imputation increased performance compared to the simple weighted versions for small sample sizes. We examined the effect of feature elimination in this setting, but the method frequently eliminated important predictors, especially in the smaller sample sizes. The results using feature elimination can be found in the supplemental material.

 \begin{figure}[htb]
 \begin{center}
 \includegraphics[width=4.5in]{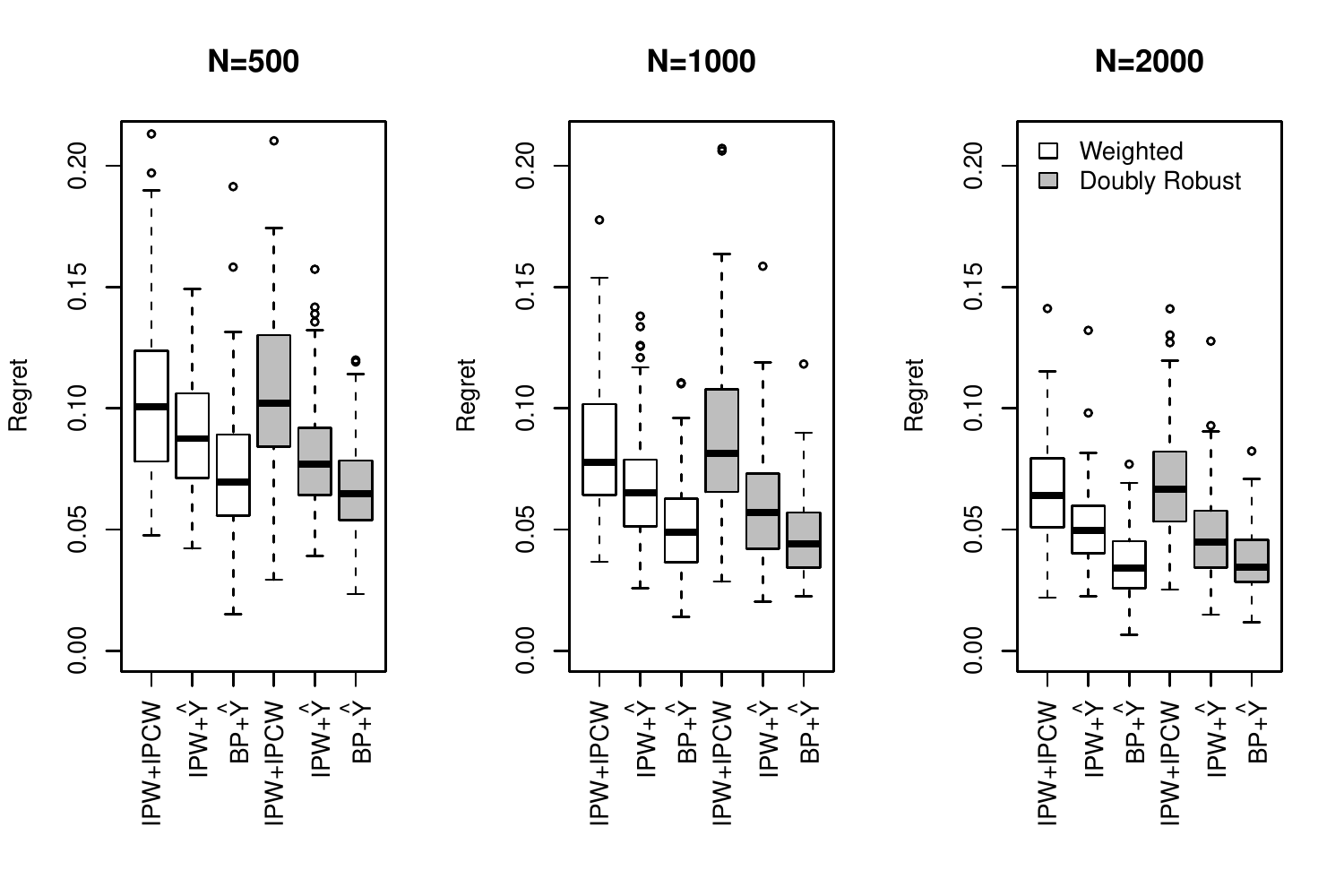}
 \end{center}
 \caption{Boxplots showing estimated population regret for setting 2. Smaller is better. IPW+IPCW: IPW combined with IPCW;  IPCW+$\hat{Y}$: IPW with imputed outcome vector; BP+$\hat{Y}$: Balanced policy with imputed outcome vector. Doubly robust versions used the correctly specified AFT model for $\hat\mu$. \label{fig:AFT}}
 \end{figure}

\section{Application to HIV Study}
\label{sec:application}

We apply the proposed method to analyze data from the UNC CFAR HIV Clinical Cohort (UCHCC). In the UCHCC study, patients were followed from initiation of an integrase strand transfer inhibitor (INSTI) in combination with at least two other antiretroviral (ARV) agents from at least one other ARV therapy class representing current standard of care for initial ART in high income clinical settings. Patients were followed until the first of ART modification or discontinuation, death, loss to follow-up,
or administrative censoring. The study data included 957 HIV-infected patients who were 72\% male, 62\% black, 29\% white, 9\% Hispanic or other races/ethnicities. The median age at treatment initiation was 44 years, 50\% were men who have sex with men (MSM) and 9\% had a history of injection drug use (IDU). At INSTI initiation (baseline), the median CD4 cell count was 514 cells/mm$^3$ (range 9 to 2970.0) and the median HIV RNA level was 1.8 log10 copies/mL (range 1.3 to 7.6). The median number of prior antiretroviral (ARV) compounds was 3 (range 0 to 17) and 29\% had no prior ARV experience. The INSTI regimen was chosen by providers and patients based on clinical indication and included in all cases one integrase strand transfer inhibitor (INSTI) and two or more nucleoside reverse-transcriptase inhibitors (NRTIs).  

The covariates include age, male (yes, no), race (black, white, or other), MSM (yes, no), IDU (yes, no), baseline CD4 count, baseline viral load (VL), an indicator if baseline RNA is undetectable, the number of prior ARV agents, and an indicator if the patient is treatment na\"ive. Categorical variables are transformed into dummy variables, resulting in an 11-dimensional covariate vector $X$. 

The treatment regimens of interest are Raltegravir (RAL), Elvitegravir (EVG), and Dolutegravir (DTG), each with 2+ NRTIs. The primary outcome of interest is time to the discontinuation of the initial INSTI regimen, which is defined as either a change in the INSTI agent or discontinuing ART for more than two weeks. The compounds under consideration received FDA approval at different times, so the maximum follow-up time is different for each regimen. The left plot in Figure \ref{fig:survivial_curves} shows the Kaplan-Meier curves for the UCHCC data by treatment over the unrestricted follow-up window. This plot indicates that there may be large treatment differences for the longer follow-up times, but we chose to limit our analysis to 2.5 years (913 days) of follow-up time. This time period was chosen because it is near the $98^\text{th}$ percentile of observed treatment discontinuations for the regimen with the least follow-up time, which ensures sufficient data coverage to support the analysis.

Among the 957 study subjects, 416 (43\%) were observed to discontinue treatment, and 319 (33\%) were censored due to loss to follow-up during the first 2.5 years. The 222 (23\%) patients still being followed after 2.5 years were administratively censored.

We applied balanced policy evaluation and learning to the UCHCC data using all available variables with the kernel defined in \eqref{eq:mRBF} and the policy class in \eqref{eq:pi_logit}. The estimated optimal treatment strategy identified by balanced policy evaluation and learning was to treat every patient with DTG and 2+ NRTIs. Each of the three INSTIs evaluated in these analyses have different barriers to resistance evolution, tolerability profiles and dosing frequency, likely affecting the durability of the regimen differentially across different patient groups. Additional analyses stratified by whether patients were ART naive at INSTI initiation or by age at INSTI initiation, supported the findings of the primary analyses, the details of which may be found in the supplemental material.

 \begin{figure}[htbp]
     \centering
     \includegraphics[width=.4\textwidth]{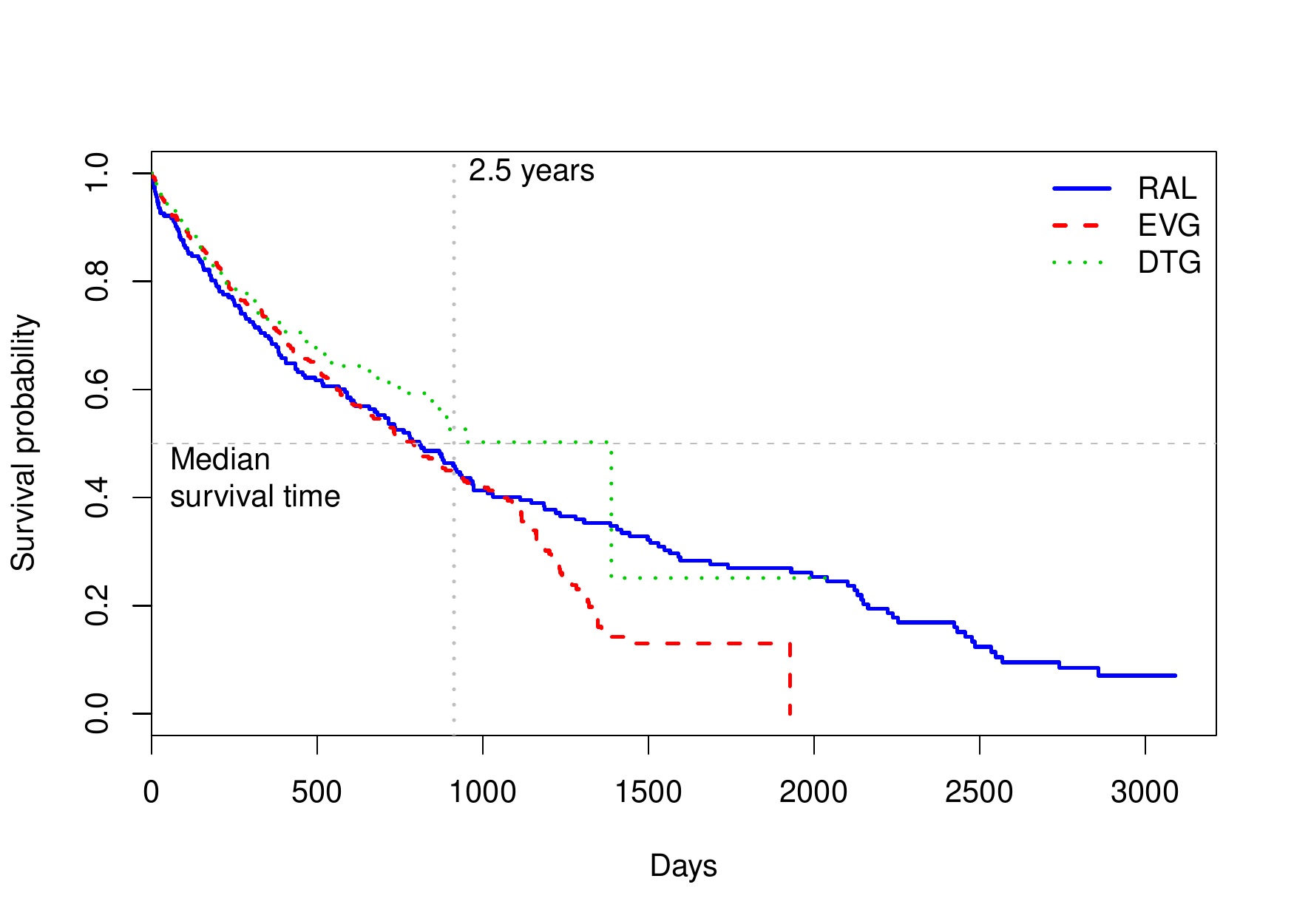}
     \includegraphics[width=.4\textwidth]{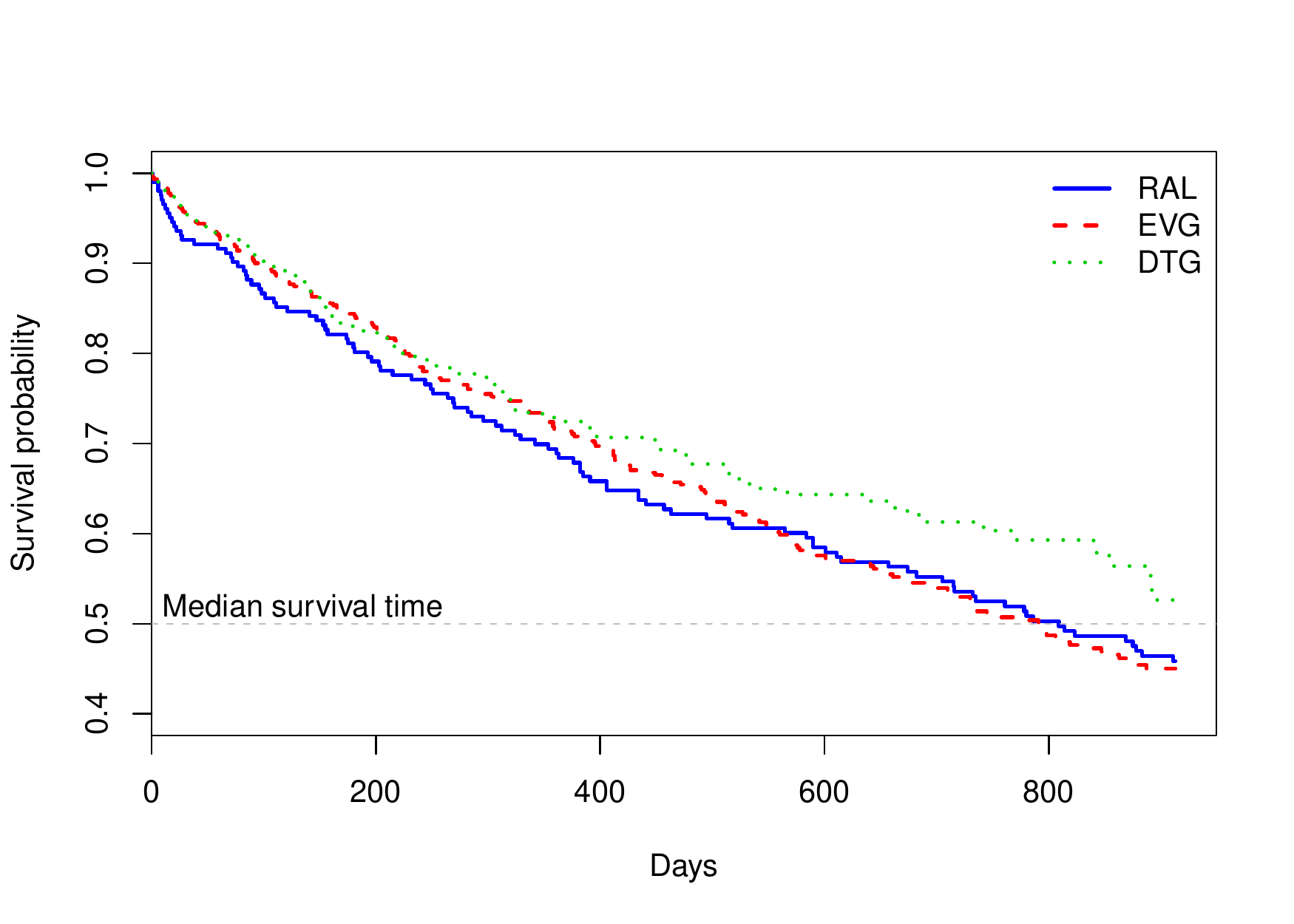} 
     \caption{Estimated marginal survival curves for each treatment. The left plot shows the entire follow-up period. The right plot shows the restricted follow-up period
     }
     \label{fig:survivial_curves}
 \end{figure}

In order to see how our method performs in comparison to existing methods on realistic data when meaningful subgroup effects are present, we modified the data to equalize baseline differences and highlight heterogeneity. Specifically, a constant scaling factor was applied to the observed times for each of the treatments so that the marginal treatment effects between any two treatments was negligible. Due to the effects of differential censoring rates between treatments, the average policy effect was estimated using the regression-based estimator \eqref{eq:reg} where the mean process $\hat{\mu}$ was estimated using RIST. 

For this analysis we used a reduced set of predictors which were identified by riskRFE as being most likely to have meaningful subgroup effects. The predictors were age, treatment na\"ive status, baseline CD4, baseline RNA, and undetectable baseline RNA indicator. The conditional expected survival times, $\widehat{Y}$, and the censoring probabilities used for IPCW were estimated using RIST with the tuning parameter settings suggested in \citet{zhu2012}. Propensity scores were estimated using a Gaussian process classifier. The balanced policy evaluation again used the kernel defined in \eqref{eq:mRBF} and the policy class in \eqref{eq:pi_logit}.

We compare the balanced policy with imputation approach to several other methods. The regression-based method as defined in equation \eqref{eq:reg} with $\hat{\mu}$ estimated using RIST, and two weighted methods as defined in equation \eqref{eq:vanilla_weighted}. The weighted methods are: IPW with IPCW using the weights defined in equation \eqref{eq:IPW_IPCW}, and IPW with imputation using $\hat{Y}$ and the weights defined in equation \eqref{eq:IPW}.

For comparison between methods, we estimate the optimal policy with a cross-validation type analysis. Specifically, the data are partitioned into ten roughly equal-sized parts. For each method under consideration, we estimate the optimal policy on nine parts of the data, and then compute several estimates of the SAPE using the remaining tenth part. By applying the above procedure, holding out a different part of the data each time, we were able to get out-of-sample predictions which should better represent expected optimal SAPE for future subjects. This approach was applied to 100 different partitions of the data and the results are presented in Table \ref{tab:Results_modified} and Figure \ref{fig:Results_modified}. 

From Table \ref{tab:Results_modified}, we observe that the proposed approach produces policies with larger mean estimated SAPE with smaller standard errors across all three evaluation criteria. The policies found by the regression-based estimator had high estimated SAPE with respect to its corresponding criteria, $\hat{\psi}^\text{Reg}(\hat{\pi})$, but low estimated SAPE with respect to the weight based criteria, $\hat{\psi}_{W^\text{IPW}}(\hat{\pi})$ and $\hat{\psi}_{W^\text{IPW,IPCW}}(\hat{\pi})$. The proposed method was able to find policies with high estimated SAPE for all criteria, indicating that the policies found were truly the best in terms of SAPE, as their improvement was insensitive to the choice of validation estimation method.


 \begin{table}[htbp]
     \centering
     \caption{Analysis of modified HIV data: Mean (sd) of estimated SAPE.}
     \begin{tabular}{lccc}   \hline \hline
     Method & $\hat{\psi}^\text{Reg}(\hat{\pi})$ & $\hat{\psi}_{W^\text{IPW}}(\hat{\pi})$ & $\hat{\psi}_{W^\text{IPW,IPCW}}(\hat{\pi})$ \\ \hline
     Reg & 667 (0.63)   & 608 (12.1) & 601 (12.7) \\
     IPW+IPCW & 658 (2.37)     & 630 (12.0) & 617 (14.1) \\
     IPW+$\hat{Y}$ & 660 (2.26)    & 626 (12.5) & 620 (15.7) \\
     BP+$\hat{Y}$ & 667 (0.43)     & 649 (6.1)  & 647 (6.9) \\ \hline
     \end{tabular} \\
     {\footnotesize Reg: Regression-based approach; IPW+IPCW: IPW combined with IPCW;  IPCW+$\hat{Y}$: IPW with} \\ \vspace*{-1mm}
     {\footnotesize imputed outcome vector; BP+$\hat{Y}$: Balanced policy with imputed outcome vector.}%
     \label{tab:Results_modified}%
 \end{table}%


 \begin{figure}[htbp]
     \centering
     \includegraphics[width=.6\textwidth]{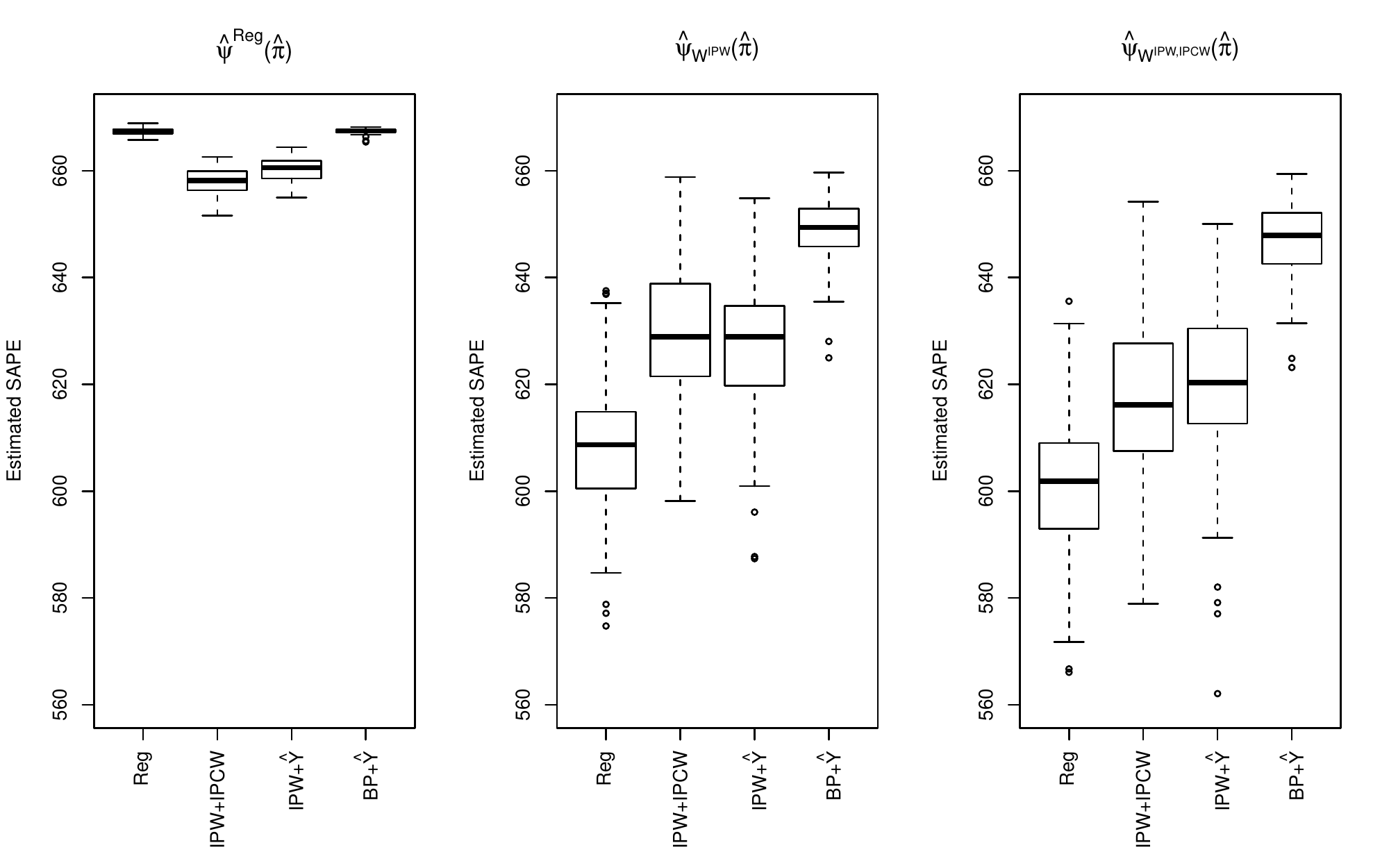}
     \caption{Boxplots showing the estimated SAPE for the modified HIV data. Larger is better. Reg: Regression-based approach; IPW+IPCW: IPW combined with IPCW;  IPCW+$\hat{Y}$: IPW with imputed outcome vector; BP+$\hat{Y}$: Balanced policy with imputed outcome vector.}
     \label{fig:Results_modified}
 \end{figure}%

 \begin{table}[htbp]
     \centering
     \caption{Analysis of modified HIV data: Mean (sd) of the percentage of subjects assigned to each treatment.}
     \begin{tabular}{lccc}   \hline \hline
     Method & RAL & EVG & DTG \\ \hline
     Reg & 42.1\% (1.62) & 51.0\% (1.75) &  6.9\% (0.77)  \\
     IPW+IPCW & 32.2\% (3.98) & 57.4\% (3.25) & 10.4\% (1.83) \\
     IPW+$\hat{Y}$ & 36.6\% (3.38) & 54.9\% (3.02) & 8.5\% (2.10) \\
     BP+$\hat{Y}$ & 48.1\% (1.17) & 49.4\% (0.94) & 2.5\% (0.73) \\ \hline
     \end{tabular} \\
     {\footnotesize Reg: Regression-based approach; IPW+IPCW: IPW combined with IPCW;  IPCW+$\hat{Y}$: IPW with} \\ \vspace*{-1mm}
     {\footnotesize imputed outcome vector; BP+$\hat{Y}$: Balanced policy with imputed outcome vector.}%
     \label{tab:Results_modified_pct}%
 \end{table}%

Table \ref{tab:Results_modified_pct} shows the mean and standard deviations of percentage of subjects that were assigned to each of the three treatment options for the different approaches over the 100 partitions of the data. The standard deviation of the treatment assignment percentages is smallest for the proposed method, indicating that it is less sensitive to small perturbations of the data when compared to the other approaches. The variability was largest for the inverse weighted methods which is likely a result of the higher variance of these estimators. The censoring rates were modest for RAL and EVG, but high for DTG. Estimating the average treatment effect in the presence of high censoring likely led to an unstable estimate and we believe that adjusting based on this estimate led to the unusually low numbers of subjects assigned to DTG.


\section{Discussion}
\label{sec:disc}

The proposed method appears to be more effective at finding optimal treatment policies, when compared to adapting previous methods to the right-censored setting. The use of unbiased weighted estimators is likely to result in high variance estimates of the SAPE which can complicate the process of finding the policy that is optimal for the population of interest. The proposed method uses a balanced policy evaluation and learning approach combined with imputation of the censoring times to reduce the variance of the estimator. Both of these approaches are likely to increase the bias, especially for small sample sizes. However, our results show that the reduction in variance offsets the increased bias and allows for more effective training of an optimal treatment policy.

The imputation approach that we propose can be used to extend any policy evaluation method to right censored data. In the simulation studies, we demonstrated that combining inverse propensity weighting with the imputation approach leads to better performance when compared with the method that uses both inverse propensity and inverse probability censoring weights. We showed that the imputation approach is compatible with any consistent estimate of the conditional survival or cumulative hazard function. We proved that the convergence rates for the conditional expected failure times is of the same order as the convergence rates for the conditional survival or cumulative hazard function estimator.

There are a number of ways in which the proposed method may be improved or extended. As currently formulated, the proposed method is limited to comparatively small datasets, both in terms of number of observations and the dimension of the covariate space. Like the original balanced policy evaluation method, the proposed method is more computationally intensive than existing approaches. There may be more efficient optimization procedures, such as alternating descent or boosting with a faster approximate method that could allow the proposed method to be applied to larger datasets. Additionally, when no parametric or semiparametric relationship can be assumed for failure time distribution, the convergence rate of the conditional expected survival times can be quite slow, especially when the number of covariates are large. We have discussed some potential solutions, but a more principled approach may be able to further increase the performance. It would also be of interest to extend our approach to dynamic treatment regimes where we wish to identify an optimal treatment sequence in a censored survival setting.

\newpage

\section{Appendix}

\begin{proof}[Proof of Lemma \ref{thm:Preservaton_of_rates}] Noting that
\begin{align*}
\fracn\sumi (1-\Delta_i) W_i(\widehat{Y}_i-T_i) 
&= \fracn\sumi (1-\Delta_i) W_i(\widehat{Y}_i-\tY_i) + \fracn\sumi (1-\Delta_i) W_i(\tY_i-T_i),
\end{align*}
we can bound each term separately. Focusing on the rate at which $\hat{Y} \rightarrow \tY$, consider the form of $\hat{Y}$ 
\begin{align*}
 \frac{\int_Y^\tau t\, d\widehat{F}(t\,|\, X,A) + \tau \widehat{S}(\tau\,|\, X,A)}{\widehat{S}(Y\,|\, X,A)} =  \frac{ - \int_Y^\tau t\, d\widehat{S}(t\,|\, X,A)}{\widehat{S}(Y\,|\, X,A)} + \frac{\tau \widehat{S}(\tau\,|\, X,A)}{\widehat{S}(Y\,|\, X,A)}.
\end{align*}
The first quantity depends on the pair $(\widehat{S}(t\,|\, X,A),\widehat{S}(Y\,|\, X,A))$ through the two maps $(A,B) \mapsto (A,1/B) \mapsto \int_y^\tau tdA/B.$
The derivative of the composition map is given by
\begin{align*}  
    (\alpha,\beta) \mapsto \left(\alpha,\frac{-\beta}{S(Y\,|\, X,A)^2}\right) \mapsto \frac{\int_y^\tau td\alpha(t)}{S(Y\,|\, X,A)} + \frac{\int_y^\tau t\,dS(Y\,|\, X,A)\beta(y)}{S^2(y\,|\, X,A)}.
\end{align*}
Similarly, the second quantity depends on the pair $(\widehat{S}(t\,|\, X,A),\widehat{S}(Y\,|\, X,A))$ through the two maps $(A,B) \mapsto (A,1/B) \mapsto \tau A/B.$
The derivative of the composition map is given by
\begin{align*}
    (\alpha,\beta) \mapsto \left(\alpha,\frac{-\beta}{S(Y\,|\, X,A)^2}\right) \mapsto \frac{\tau \alpha(\tau)}{S(Y\,|\, X,A)} + \frac{\tau S(y\,|\, X,A)\beta(y)}{S^2(Y\,|\, X,A)}.
\end{align*}
By the Hadamard differentiability of the maps, Theorems 2.8 and 12.1 of \citet{kosorok2008}, and Assumption \ref{as:convergence_Y} we have that $\fracn\sumi \widehat{Y}_i-\tY_i = O_p(r_n)$. 

\noindent Returning to $\fracn\sumi (1-\Delta_i) W_i(\widehat{Y}_i-\tY_i)$, by Cauchy-Schwarz we have
\begin{align*}
    \fracn\sumi (1-\Delta_i) W_i(\widehat{Y}_i-\tY_i) 
    &\leq \left(\fracn\|W\|_2^2\right)^{\frac{1}{2}}\left(\fracn\sumi (\widehat{Y}_i-\tY_i)^2\right)^{\frac{1}{2}}.
\end{align*}
The proof of \cite{kallus2017} Theorem 3 implies that $\fracn\|W\|_2^2=O_p(1)$, and thus

$$\left|\fracn \sumi (1-\Delta_i)W_i(\widehat{Y}_i-\tY_i)\right| \leq \left(\fracn\|W\|^2_2\right)^{\frac{1}{2}}\left( \fracn \sumi (1-\Delta_i)(\widehat{Y}_i-\tY_i)^2 \right)^\frac{1}{2}.$$

\noindent By assumption \ref{as:leave_one_out} and the exchangeability of the data, we have
\begin{align*}
E\left[ \fracn \sumi (1-\Delta_i)(\widehat{Y}_i-\tY_i)^2 \right] 
&\leq E\left[ (1-\Delta_n)(\widehat{Y}^{(n)}_n-\tY_n)^2 \right] .
\end{align*}
Now,
\begin{align*}
    &E\left[ (1-\Delta_i)(\widehat{Y}^{(n)}_n-\tY_n)^2 \right] \\
    &\leq \tau E\left[ \left|\widehat{E}_{n-1}(T\,|\, X_n,A_n,T>Y_n)-E(T\,|\, X_n,A_n,T>Y_n)\right| \wedge \tau \right]  
\end{align*}
(since both $\widehat{E}_{n-1}[T]$ and $E[T]$ are constrained to be $\in(0,\tau]$)
\begin{align*}
    &= \tau E\left[ \left|\widehat{E}_{n-1}(T\,|\, X,A,T>Y)-E(T\,|\, X,A,T>Y)\right| \wedge \tau \right] \\ 
    &= \tau E\left[\left|\frac{-2\tau\widehat{S}_{n-1}(\tau\,|\, X,A) + Y\widehat{S}_{n-1}(Y\,|\, X,A) + \int_Y^\tau \widehat{S}_{n-1}(t\,|\, X,A)\,dt}{\widehat{S}_{n-1}(Y\,|\, X,A)} \right. \right.\tag{$*1$} \\
    &\pushright{ + \left. \left. \frac{2\tau S(\tau\,|\, X,A) - Y S(Y\,|\, X,A) - \int_Y^\tau  S(t\,|\, X,A)\,dt}{S(Y\,|\, X,A)}\right| \wedge \tau \right]. } 
\end{align*}
By adding and subtracting $$\frac{2\tau S(\tau\,|\, X,A) - Y S(Y\,|\, X,A) - \int_Y^\tau  S(t\,|\, X,A)\,dt}{\widehat{S}_{n-1}(Y\,|\, X,A)},$$ combining terms and simplifying, we have
\begin{align}
    (*1)
    &\leq 8\tau^2 E\left[ \left|\frac{\sup_{t\leq\tau}\left(\widehat{S}_{n-1}(t\,|\, X,A) - S(t\,|\, X,A)\right)}{\widehat{S}_{n-1}(\tau\,|\, X,A)} \right| \wedge\tau \right]. \tag{$*2$}
\end{align}
If we assume that $\inf_{X,A} S(\tau\,|\, X,A) \geq \delta \text{ for some } \delta > 0,$
then,
\begin{align*}
(*2)&\leq \frac{16\tau^3}{\delta} \left[ \sup_{t\leq\tau}\left|\widehat{S}_{n-1}(t\,|\, X,A) - S(t\,|\, X,A)\right|  +\text{Pr}\left( \widehat{S}_{n-1}(t\,|\, X,A) \leq \frac{\delta}{2}\right) \right]. \tag{$*3$}
\end{align*}
But 
\begin{align*}
    \text{Pr}\left( \widehat{S}_{n-1}(t\,|\, X,A) \leq \frac{\delta}{2} \right)
    &\leq \text{Pr}\left( \widehat{S}_{n-1}(t\,|\, X,A) - S(t\,|\, X,A) \leq - \frac{\delta}{2} \right) \\
    &\leq \frac{E \left|\sup_{t\leq\tau}\left(\widehat{S}_{n-1}(t\,|\, X,A) - S(t\,|\, X,A)\right) \right| }{\delta/2},
\end{align*}
thus
\begin{align*}
    (*3) \leq \frac{8\tau^3}{\delta} \left( 1+\frac{2}{\delta} \right)E \left|\sup_{t\leq\tau}\left(\widehat{S}_{n-1}(t\,|\, X,A) - S(t\,|\, X,A)\right) \right|.
\end{align*}
Note that $r_n \rightarrow 0$ iff $r_{n-1} \rightarrow 0$, so the rates should be similar.
Turning our attention to $\fracn\sumi (1-\Delta_i) W_i(\tY_i-T_i)$. Let $F_i=\sigma(X_i,A_i,Y_i,\Delta_i)$ be the sigma field for a single observation and  $\overline{F}_n=\sigma(\bigcup\limits_{1\leq i \leq n} F_i)$. For any $i$, note that
\begin{align*}
\mathbb{E}[(1-\Delta_i) W_i(\tY_i-T_i)] &= \mathbb{E}\left[\mathbb{E}[(1-\Delta_i) W_i(\tY_i-T_i)\,|\, \overline{F}_n]\right] \\
&= \mathbb{E}\left[(1-\Delta_i) W_i\left(\mathbb{E}[T_i\,|\, X_i,A_i,Y_i]-\mathbb{E}[T_i\,|\, X_i,A_i,Y_i]\right)\right] = 0.
\end{align*}
For $1\leq i\neq j \leq n$, also note that
\begin{align*}
&\mathbb{E}[(1-\Delta_i) W_i(\tY_i-T_i)(1-\Delta_j) W_j(\tY_j-T_j)] \\
&=\mathbb{E}\left[(1-\Delta_i)(1-\Delta_j)W_i W_j\mathbb{E}[(\tY_i-T_i)(\tY_j-T_j)\,|\, \sigma(F_i\cup F_j)]\right],
\end{align*}
but,
\begin{align*}
&\mathbb{E}[(\tY_i-T_i)(\tY_j-T_j)\,|\, \sigma(F_i\cup F_j)] \\
&=\mathbb{E}\left[ \mathbb{E}[(\tY_i-T_i)\,|\, X_i,A_i,Y_i,\Delta_i]\mathbb{E}[(\tY_j-T_j)\,|\, X_j,A_j,Y_j,\Delta_j] \right]=0.
\end{align*}
This implies that $\fracn\sumi (1-\Delta_i) W_i(\tY_i-T_i)$ is approximately zero with variance 
\begin{align*}
    \fracn\left[\fracn \sumi \mathbb{E}\left[ (1-\Delta_i)W_i^2(\tY-T)^2 \right]\right] = \fracn \mathbb{E}\left[ (1-\Delta_i)W_i^2(\tY-T)^2 \right] \rightarrow 0,
\end{align*}
by exchangeability and by the fact that $\limsup_{n\rightarrow\infty} \mathbb{E}[ (1-\Delta_i)W_i^2(\tY-T)^2]<\infty$. Combining these results, we obtain 
\begin{align*}
 \frac{1}{n}\sum_{i=1}^n (1-\Delta_i) W_i(\widehat{Y}_i-T_i) &= O_p(r_n). \tag*{\qedhere}
\end{align*}
\end{proof}

\newpage

\bibliographystyle{agsm}

\bibliography{main}

\end{document}



\if1\blind
{
  \title{\bf Supplemental Material for ``Balanced Policy Evaluation and Learning for Right Censored Data''}
  \author{Owen E. Leete\hspace{.2cm}\\
    Department of Biostatistics, \\ University of North Carolina at Chapel Hill\\
    and \\
    Nathan  Kallus \\
    School of Operations Research and Information Engineering,\\ Cornell University\\
    and \\
    Michael G. Hudgens \\
    Department of Biostatistics, \\ University of North Carolina at Chapel Hill\\
    and \\
    Sonia Napravnik \\
    Division of Infectious Diseases, Department of Medicine, \\ University of North Carolina at Chapel Hill \\
    and \\
    Michael R. Kosorok \\
    Department of Biostatistics, \\ University of North Carolina at Chapel Hill}
  \maketitle
  \newpage
} \fi

\def\spacingset#1{\renewcommand{\baselinestretch}%
{#1}\small\normalsize} \spacingset{1}

\if0\blind
{
  \bigskip
  \bigskip
  \bigskip
  \begin{center}
    {\LARGE\bf Supplemental Material for ``Balanced Policy Evaluation and Learning for Right Censored Data''}
\end{center}
  \medskip
} \fi

\subsection*{Additional Simulation Study Results}

Figure \ref{fig:BP_Orig2} shows the results for simulation setting 1 for the simple weighted and doubly robust estimators. For this setting, the mean was a complicated function of the covariates. The nonparametric estimate of the mean was not accurate enough for the doubly robust methods over the corresponding simple weighted versions.

\begin{figure}[htbp]
\begin{center}
\includegraphics[width=4.5in]{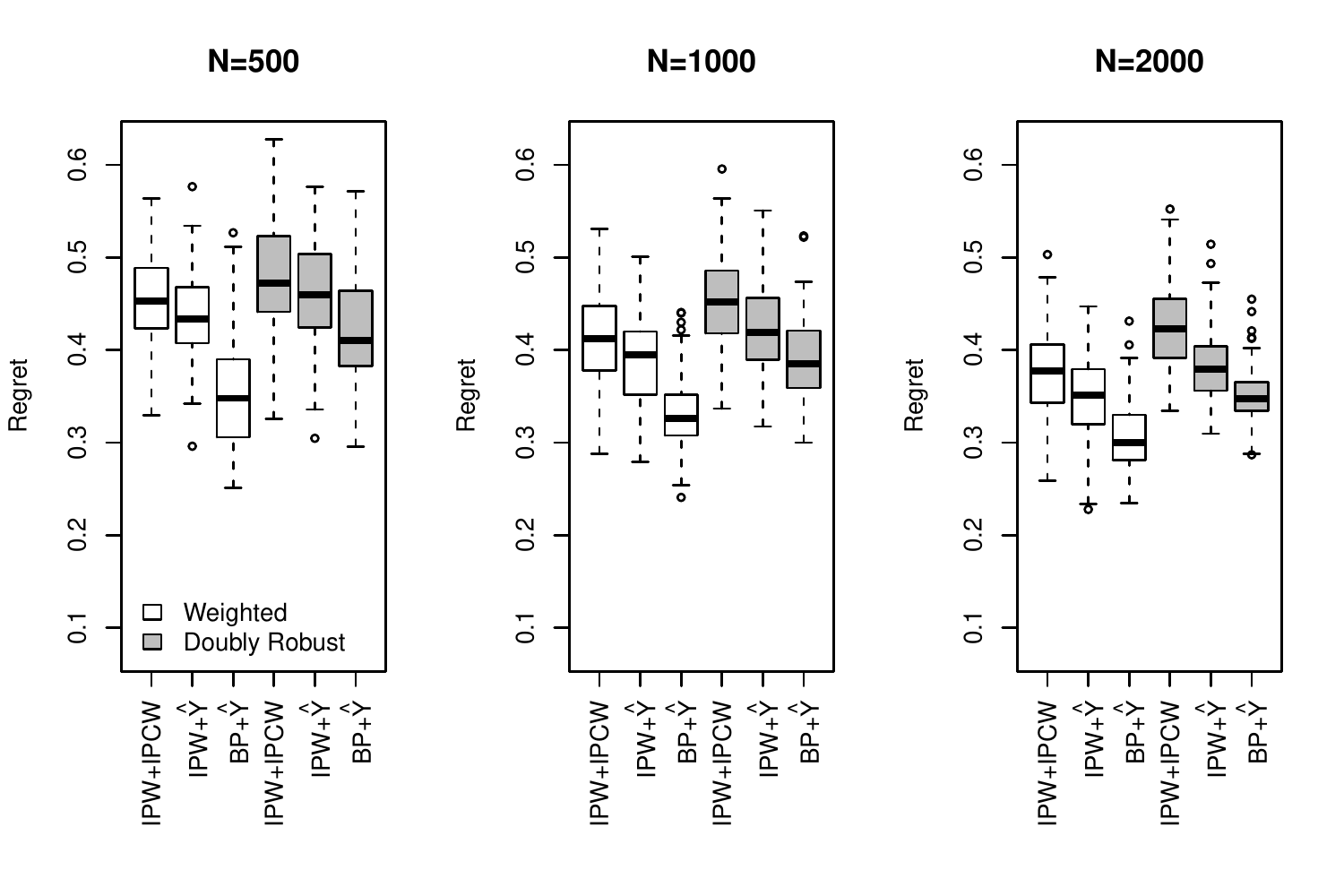}
\end{center}
\caption{Boxplots showing estimated population regret for setting 1. Smaller is better. IPW+IPCW: IPW combined with IPCW;  IPCW+$\hat{Y}$: IPW with imputed outcome vector; BP+$\hat{Y}$: Balanced policy with imputed outcome vector. Doubly robust versions used the correctly specified AFT model for $\hat\mu$. \label{fig:BP_Orig2}}
\end{figure}

Figure \ref{fig:AFT2} shows the results for simulation setting 2 with and without variable selection. For the smaller sample sizes, the variable selection method used was unable to reliably identify the meaningful variables. For the largest sample size, variable selection was able to improve the performance of the policy learners, which is most clearly seen in the IPW+$\hat{Y}$ and IPW+IPCW estimators.

\begin{figure}[htbp]
\begin{center}
\includegraphics[width=4.5in]{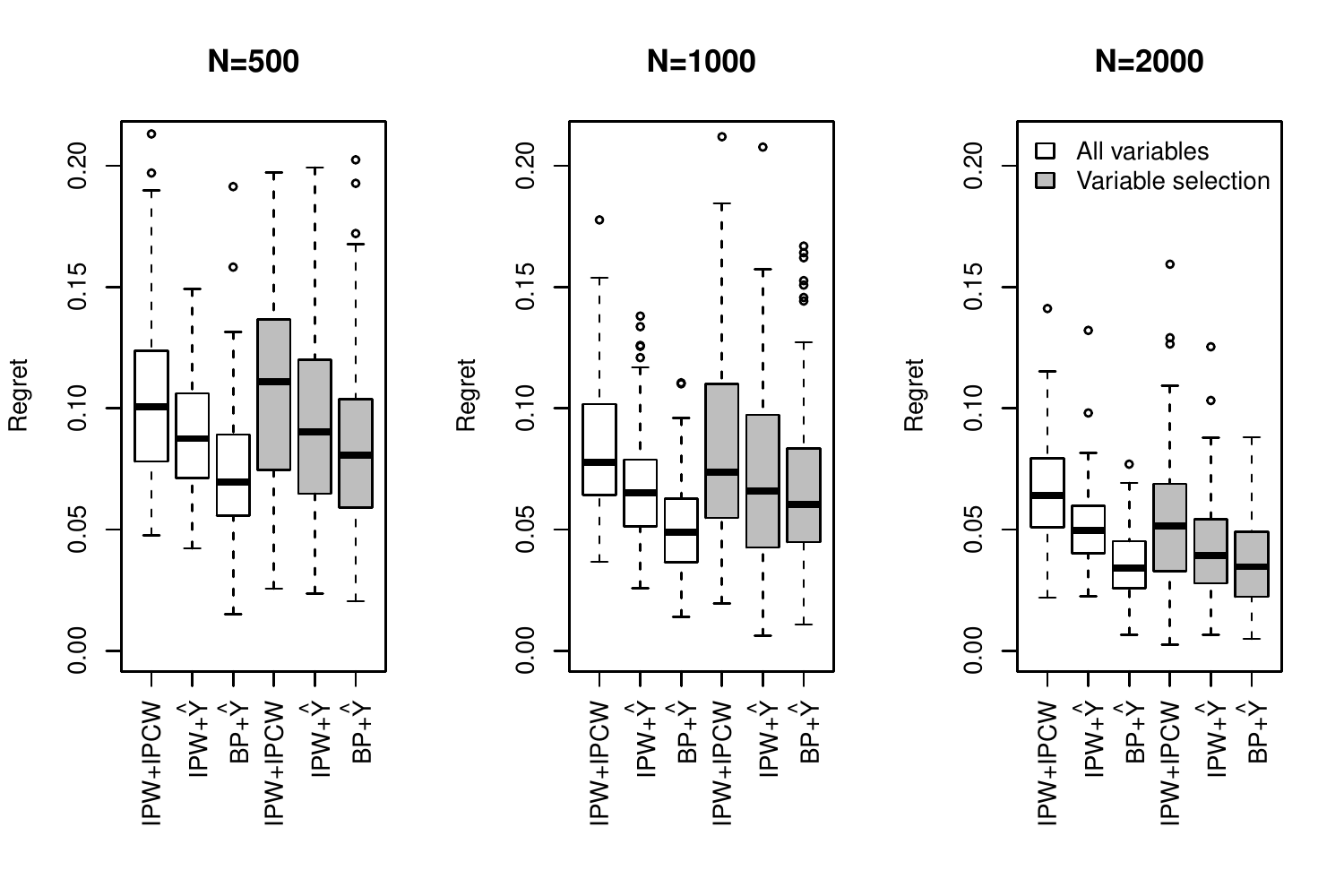}
\end{center}
\caption{Boxplots showing estimated population regret for setting 2. Smaller is better. IPW+IPCW: IPW combined with IPCW;  IPCW+$\hat{Y}$: IPW with imputed outcome vector; BP+$\hat{Y}$: Balanced policy with imputed outcome vector. Doubly robust versions used the correctly specified AFT model for $\hat\mu$. \label{fig:AFT2}}
\end{figure}

\subsection*{Additional Application Results}

The differences between the different INSTIs during the first 2.5 years of followup were quite small. A larger difference between RAL and EVG does not start to appear in the Kaplan-Meier plots until about 3 years of follow-up, indicating that there may be some effects that do not manifest until later in the treatment course. 

Based on clinical experience, we believe that treatment na\"ive patients should respond differently than treatment experienced patients and younger patients should respond differently than older patients. 

Figure \ref{fig:survivial_curves_appendix} shows survival curves broken down by treatment and treatment experience, and by treatment and age. For both treatment na\"ive and treatment experienced patients, Dolutegravir is associated with longer survival times. Similarly, for patients aged 30 years and younger as well as patients over 30, Dolutegravir appears to be superior to the other treatment options.

\begin{figure}[htbp]
    \centering
    \includegraphics[width=.4\textwidth]{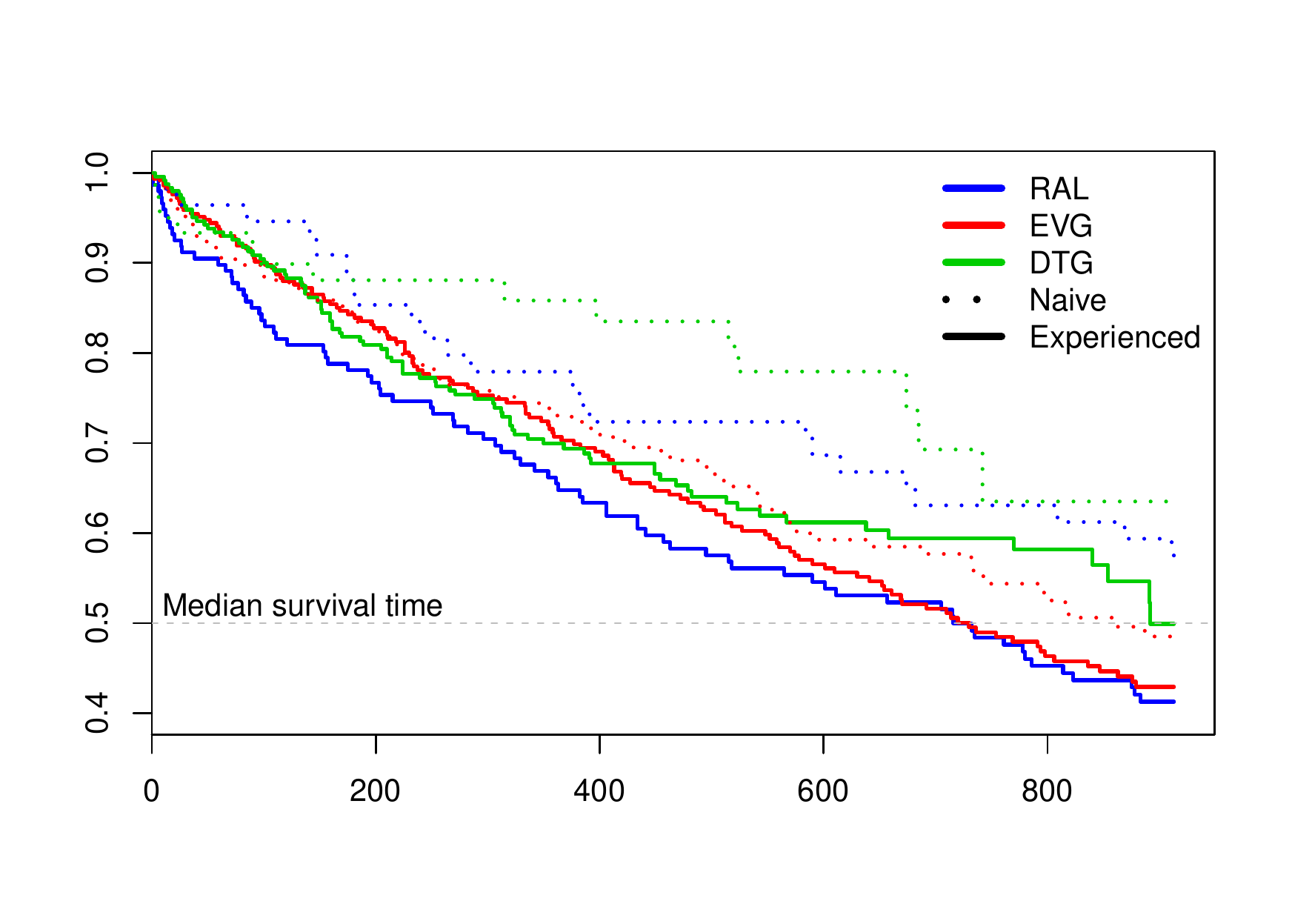}
    \includegraphics[width=.4\textwidth]{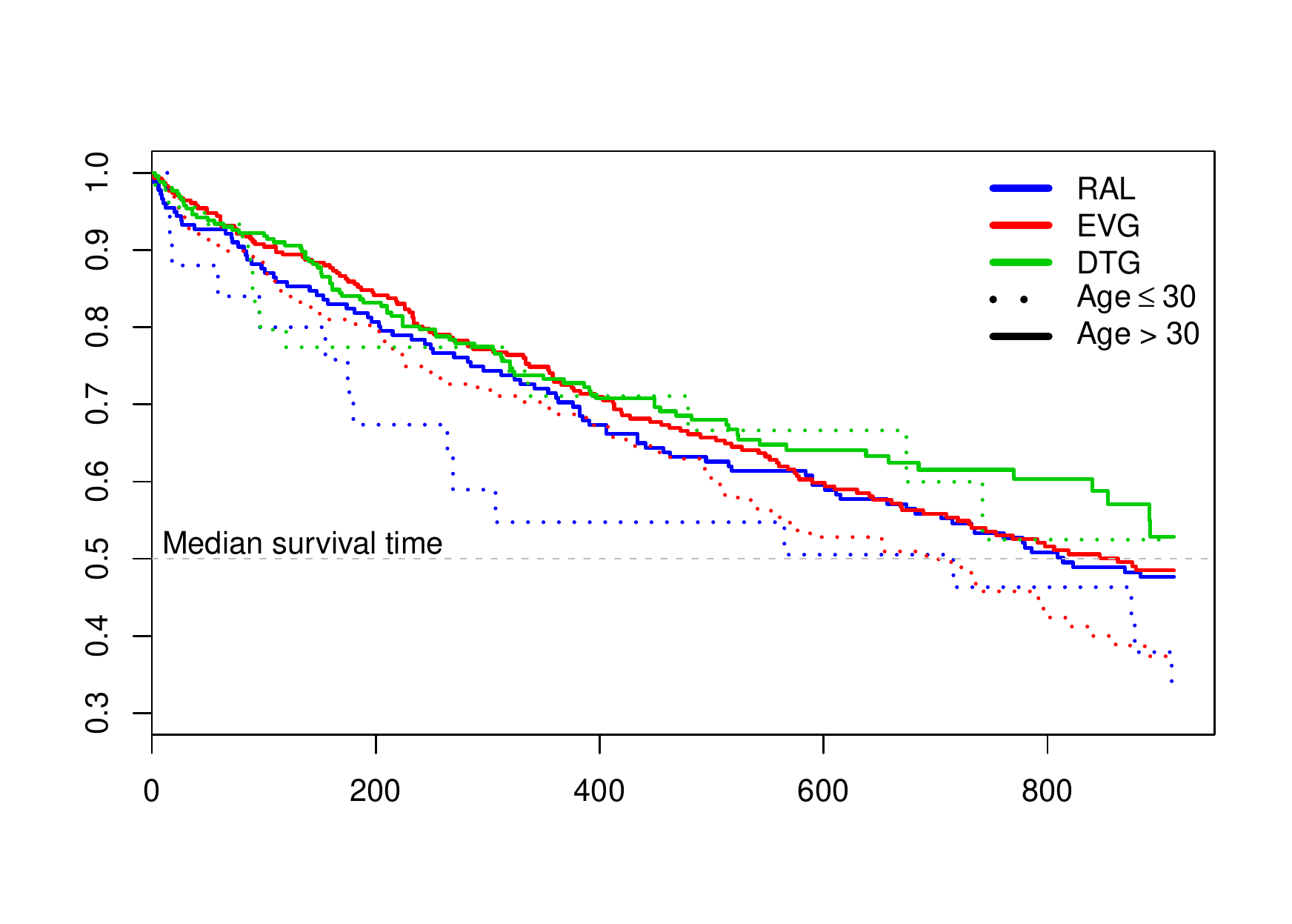}
    \caption{The left plot shows the survival curves for each treatment broken down by treatment na\"ive vs. experienced. The  right plot shows the survival curves for each treatment broken down by age$\leq$30 vs age$>30$. In both cases Dolutegravir was associated with longer survival times in each subpopulation.
    }
    \label{fig:survivial_curves_appendix}
\end{figure} 

Figure \ref{fig:survivial_curves_modified_appendix} shows the average projected survival curves under several ITR estimation methods, as well as the observed survival curve for the modified data. The regression based method resulted in a projected survival rate which is less than the observed survival rate in the data. The other methods examined resulted in improved projected survival rates when compared to the observed rate, but the best projected survival rates were associated with balanced policy evaluation with imputation.

\begin{figure}[htbp]
    \centering
    \includegraphics[width=.6\textwidth]{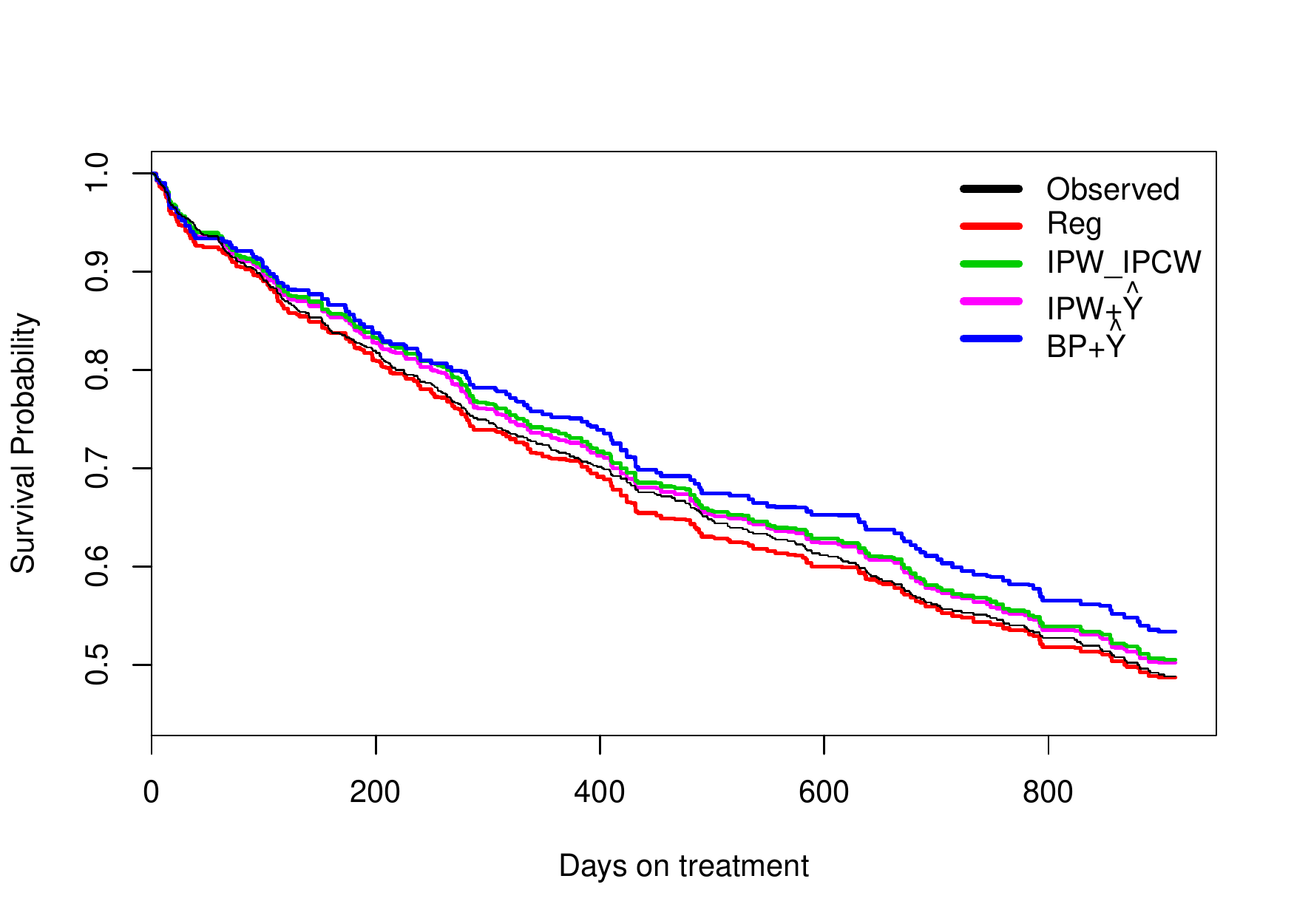}
    \caption{Average projected survival curves under several ITR estimation approaches compared to the observed survival curve. Reg: Regression-based approach; IPW+IPCW: IPW combined with IPCW;  IPCW+$\hat{Y}$: IPW with imputed outcome vector; BP+$\hat{Y}$: Balanced policy with imputed outcome vector. 
    }
    \label{fig:survivial_curves_modified_appendix}
\end{figure}